\documentclass[runningheads,a4paper,envcountsect,envcountsame]{llncs}

\usepackage{amssymb}
\usepackage{graphicx}
\usepackage{url}

\usepackage{amsmath}
\usepackage{fontawesome}
\usepackage{mathpartir}
\usepackage{stmaryrd}
\usepackage[all]{xy}
\SelectTips{cm}{} 
 
\newcommand{\keywords}[1]{\par\addvspace\baselineskip
\noindent\keywordname\enspace\ignorespaces#1}

\newcommand{\abort}{\mathsf{abort}}
\renewcommand{\Box}{\square}
\renewcommand{\case}[5]{\mathsf{case}\,#1\,\mathsf{of}\,#2.#3;#4.#5}
\newcommand{\casebare}{\mathsf{case}}
\newcommand{\Cat}{\mathcal{C}}
\newcommand{\defeq}{\triangleq}
\newcommand{\den}[1]{\llbracket#1\rrbracket}
\newcommand{\dia}{\ensuremath{\mathsf{dia}}}
\newcommand{\Dia}{\blacklozenge}
\newcommand{\early}{\ensuremath{\mathsf{early}}}
\newcommand{\IKd}{IK_{\Dia}}
\newcommand{\letd}[3]{\mathsf{let\,dia}\,#1\,\mathsf{be}\,#2\,\mathsf{in}\,#3}
\newcommand{\letdbare}{\mathsf{let\,dia}}
\newcommand{\lock}{\mbox{\faUnlock}}

\newcommand{\open}{\ensuremath{\mathsf{open}}}
\newcommand{\red}{\leadsto}
\newcommand{\RED}[1]{RED_{#1}}
\newcommand{\shut}{\ensuremath{\mathsf{shut}}}
\newcommand{\sslock}{\scriptsize\mbox{\faUnlock}}
\newcommand{\tlock}{\tiny\mbox{\faUnlock}}

\begin{document}

\mainmatter  

\title{Fitch-Style Modal Lambda Calculi}

\titlerunning{Fitch-Style Modal Lambda Calculi}

%
%
\author{Ranald Clouston%
\thanks{We gratefully acknowledge discussions with Patrick Bahr, Lars Birkedal,
Ale\v{s} Bizjak, Christian Uldal Graulund, G.A. Kavvos, Bassel Mannaa,
Rasmus Ejlers M{\o}gelberg, Andrew M. Pitts, and Bas Spitters, and the comments of the
anonymous referees. This research was supported by a research grant (12386) from
Villum Fonden.}%
}
\authorrunning{Fitch-Style Modal Lambda Calculi}

\institute{Department of Computer Science, Aarhus University, Denmark \\
\url{ranald.clouston@cs.au.dk}
}

%
%

\maketitle

\begin{abstract}
Fitch-style modal deduction, in which modalities are eliminated by opening a subordinate proof, and introduced by shutting one, were investigated in the 1990s as a basis for lambda calculi. We show that such calculi have good computational properties for a variety of intuitionistic modal logics. Semantics are given in cartesian closed categories equipped with an adjunction of endofunctors, with the necessity modality interpreted by the right adjoint. Where this functor is an idempotent comonad, a coherence result on the semantics allows us to present a calculus for intuitionistic S4 that is simpler than others in the literature. We show the calculi can be extended \`{a} la tense logic with the left adjoint of necessity, and are then complete for the categorical semantics.
\end{abstract}


\section{Introduction}

The Curry-Howard propositions-as-types
isomorphism~\cite{Girard:Proofs,Sorensen:Lectures,Wadler:Propositions} provides a
correspondence between natural deduction and typed lambda calculus of interest to
both logicians and computer scientists.
For the logician, term assignment offers a convenient notation to express and reason
about syntactic properties such as proof normalisation, and, especially in the
presence of dependent types, allows proofs of non-trivial mathematical theorems to
be checked by computer programs.
For the computer scientist, logics have been repurposed as typing disciplines to
address problems in computing in sometimes surprising ways.
Following Lambek~\cite{Lambek:Introduction}, categories form a third leg of the
isomorphism.
Categorical semantics can be used to prove the consistency of a calculus, and
they are crucial if we wish to prove or program in some particular mathematical
setting.
For example, see the use of the topos of trees as a setting for both programming
with guarded recursion, and proof by L\"{o}b induction, by Clouston et
al~\cite{Clouston:Guarded}.

This work involved two functors, `later' and `constant'.
Where functors interact appropriately with finite
products they correspond to necessity modalities in intuitionistic normal modal logic,
usually written $\Box$.
Such modalities have been extensively studied by logicians, and the corresponding
type-formers are widely applicable in computing, for example to
monads~\cite{Moggi:Computational}, staged programming~\cite{Davies:Modal},
propositional truncation~\cite{Awodey:Propositions}, and recent work in
homotopy type theory~\cite{Shulman:Brouwer}.
There is hence a need to develop all sides of the Curry-Howard-Lambek isomorphism
for necessity modalities.
Approaches to modal lambda calculi are diverse; see the survey by
Kavvos~\cite{Kavvos:Many}, and remarks in the final section of this paper. This paper
focuses on \emph{Fitch-style} modal lambda calculi as first proposed by
Borghuis~\cite{Borghuis:Coming} and (as the ``two-dimensional'' approach) by
Martini and Masini~\cite{Martini:Computational}.

Fitch-style modal lambda calculi%
\footnote{`Fitch-style' deduction can also be used to mean the linear presentation
of natural deduction with subordinate proofs for implication.}
adapt the proof methods of Fitch~\cite{Fitch:Symbolic} in which given a formula
$\Box A$ we may open a `(strict) subordinate proof' in which we eliminate the
$\Box$ to get premise $A$. Such a subordinate proof with conclusion $B$ can then be
shut by introducing a $\Box$ to conclude $\Box B$. Different
modal logics can be encoded by tweaking the open and shut rules; for example we
could shut the proof to conclude merely $B$, if we had the T axiom $\Box B\to B$.
Normal modal logics are usually understood with respect to Kripke's possible worlds
semantics (for the intuitionistic version, see e.g.
Simpson~\cite[Section 3.3]{Simpson:Proof}). In this setting Fitch's approach is highly
intuitive, as opening a subordinate proof corresponds to travelling to a generic related
world, while shutting corresponds to returning to the original world. See
Fitting~\cite[Chapter 4]{Fitting:Proof} for a lengthier discussion of this approach to
natural deduction.

Borghuis~\cite{Borghuis:Coming} kept track of subordinate proofs in a sequent
presentation by introducing a new structural connective to the context when a $\Box$
is eliminated, and removing it from the context when one is introduced, in a style
reminiscent of the treatment of modal logic in display
calculus~\cite{Wansing:Sequent}, or for that matter of the standard duality between
implication and comma. To the category theorist, this suggests an operation on
contexts \emph{left adjoint} to $\Box$.
This paper exploits this insight by presenting categorical semantics for
Fitch-style modal calculi for the first time, answering the challenge of de Paiva and
Ritter~\cite[Section 4]{dePaiva:Basic}, by modelling necessity modalities as right
adjoints. This is logically sound and complete, yet less general than modelling
modalities as monoidal functors as done for example by Bellin et
al.~\cite{Bellin:Extended}. For example, truncation in sets is monoidal but has no
right adjoint. Nonetheless adjunctions are ubiquitous, and in their presence we
argue that the case for Fitch-style calculi is compelling. Examples of right adoints
of interest to type theorists include the aforementioned modalities of
guarded recursion, the closure modalities of (differential) cohesive
$\infty$-toposes~\cite[Section 3]{Schreiber:Differential}, and atom-abstraction in
nominal sets~\cite{Menni:About}.

In Section~\ref{sec:IK} we present Borghuis's calculus for the logic Intuitionistic K,
the most basic intuitionistic modal logic of necessity. To the results of confluence,
subject reduction, and strong normalisation already shown by Borghuis we add
canonicity and the subformula property, with the latter proof raising a subtle issue
with sums not previously observed. We give categorical semantics for this style of
calculus for the first time and prove soundness. In Section~\ref{sec:left} we
introduce the left adjoint as a first-class type former \`{a} la intuitionistic tense
logic~\cite{Ewald:Intuitionistic}, in which the ``everywhere in the future'' modality is
paired with ``somewhere in the past''. To our knowledge this is the first natural
deduction calculus, let alone lambda calculus, for any notion of tense logic. It is not
entirely satisfactory as it lacks the subformula property, but it does allow us to prove
categorical completeness. In Section~\ref{sec:S4} we show how the basic
techniques developed for Intuitionistic K extend to Intuitionistic S4, one of the
most-studied intuitionistic modal logics. Instead of working with known Fitch-style
calculi for this logic~\cite{Pfenning:Modal,Davies:Modal} we explore a new,
particularly simple, calculus where the modality is \emph{idempotent}, i.e. $\Box A$
and $\Box\Box A$ are not merely logically equivalent, but isomorphic. Our semantics
for this calculus rely on an unusual `coherence' proof. In Section~\ref{sec:IR}
we present a calculus corresponding to the logic Intuitionistic R.
In Section~\ref{sec:concl} we conclude with a discussion of related and
further work.


\section{Intuitionistic K}\label{sec:IK}

This section presents results for the calculus of Borghuis~\cite{Borghuis:Coming}
for the most basic modal logic for necessity, first identified to our knowledge by
Bo{\v{z}}i{\'c} et al.~\cite{Bozic:Models} as $HK_{\Box}$; following
Yokota~\cite{Yokota:General} we use the name Intuitionistic K (IK). This logic
extends intuitionistic logic with a new unary connective $\Box$, one new axiom

  K: $\;\Box(A\to B)\to\Box A\to\Box B$ \\
and one new inference rule

  \emph{Necessitation: } if $A$ is a theorem, then so is $\Box A$.

\subsection{Type System}\label{Sec:IKtype}

Contexts are defined by the grammar
\[
  \Gamma \defeq \cdot \mid \Gamma,x:A  \mid \Gamma,\lock
\]
where $x$ is a variable not in $\Gamma$, $A$ is a formula of intuitionistic modal
logic, and $\lock$ is called a \emph{lock}. The open lock symbol is used to suggest
that a box has been opened, allowing access to its contents.

Ignoring variables and terms, sequents $\Gamma\vdash A$ may be interpreted as
intuitionistic modal formulae by the translation
\begin{itemize}
  \item $\den{\cdot\vdash A} = A$;
  \item $\den{B,\Gamma\vdash A} = B\to\den{\Gamma\vdash A}$;
  \item $\den{\lock,\Gamma\vdash A} = \Box\den{\Gamma\vdash A}$.
\end{itemize}
This interpretation will suffice to confirm the soundness and completeness of our
calculus, considered as a natural deduction calculus, with respect to $IK$.
It is however not a satisfactory basis for a categorical semantics, because it does not
interpret the context as an object. In Section~\ref{subsec:IKcat} we shall see that
$\lock$ may instead by interpreted as a \emph{left adjoint} of $\Box$, applied to the
context to its left.

Figure~\ref{fig:IKtyping} presents the typing rules. Rules for the product
constructions
$1$, $A\times B$, $\langle\rangle$, $\langle t,u\rangle$, $\pi_1\,t$, $\pi_2\,t$
are as usual and so are omitted, while sums are discussed at
the end of Section~\ref{sec:IKcomp}. Note that variables can only be introduced or
abstracted if they do not appear to the left of a lock. In the variable rule the context
$\Gamma'$ builds in variable exchange, while in the $\open$ rule $\Gamma'$ builds
in variable weakening. Exchange of variables with locks, and weakening for
locks, are not admissible.

\begin{figure}[t]
  \begin{mathpar}
    \inferrule*[right=$\lock\notin\Gamma'$]{ }{\Gamma, x : A,\Gamma' \vdash x : A}
    \and
    \inferrule*{
      \Gamma,x: A \vdash t:B}{
      \Gamma \vdash \lambda x.t:A\to B}
    \and
    \inferrule*{
      \Gamma \vdash t:A\to B \\
      \Gamma \vdash u:A}{
      \Gamma \vdash t\,u:B}
    \and
    \inferrule*{
      \Gamma,\lock \vdash t:A}{
      \Gamma\vdash \shut\,t:\Box A}
    \and
    \inferrule*[right=$\lock\notin\Gamma'$]{
      \Gamma \vdash t:\Box A}{
      \Gamma,\lock,\Gamma' \vdash \open\,t: A}
  \end{mathpar}
  \caption{Typing rules for Intuitionistic K}
  \label{fig:IKtyping}
\end{figure}

\begin{theorem}[Logical Soundness and Completeness]\label{thm:IKlogsc}
A formula is a theorem of $IK$ if and only if it is an inhabited type in the empty
context.
\end{theorem}

We can for example show that the $K$ axiom is inhabited:
\[
\inferrule*{
  f:\Box(A\to B),x:\Box A,\lock\vdash\open\,f: A\to B \\
  f,x,\lock\vdash\open\,x: A
}
{
\inferrule*{
  f:\Box(A\to B),x:\Box A,\lock\vdash(\open\,f)(\open\,x): B
}
{
  f:\Box(A\to B),x:\Box A\vdash\shut((\open\,f)(\open\,x)):\Box B
}}
\]

\subsection{Computation}\label{sec:IKcomp}

We extend the usual notion of $\beta$-reduction on untyped terms with the rule
\[
  \open\,\shut\,t \;\mapsto\; t
\]
We write $\red$ for the reflexive transitive closure of $\mapsto$. This relation is
plainly confluent. Two lemmas, proved by easy inductions on the derivation of the
terms $t$, then allow us to prove subject reduction:

\begin{lemma}[Variable Weakening]\label{lem:weak}
If $\Gamma,\Gamma'\vdash t:B$ then $\Gamma,x:A,\Gamma'\vdash t:B$.
\end{lemma}

\begin{lemma}[Substitution]\label{lem:subs}
If $\Gamma,x:A,\Gamma'\vdash t:B$ and $\Gamma\vdash u:A$ then
$\Gamma,\Gamma'\vdash t[u/x]:B$.
\end{lemma}

\begin{theorem}[Subject Reduction]\label{lem:red}
If $\Gamma\vdash t:A$ and $t\mapsto u$ then $\Gamma\vdash u:A$.
\end{theorem}
\begin{proof}
$\beta$-reduction for $\to$ requires Lemma~\ref{lem:subs}, and for $\Box$ requires
Lemma~\ref{lem:weak}.
\end{proof}

A term $t$ is \emph{normalisable} if there exists an integer $\nu(t)$ bounding the
length of any reduction sequence starting with $t$, and \emph{normal} if $\nu(t)$ is
$0$. By standard techniques we prove the following theorems:

\begin{theorem}[Strong Normalisation]\label{thm:sn}
Given $\Gamma\vdash t:A$, the term $t$ is normalisable.
\end{theorem}

\begin{theorem}[Canonicity]\label{thm:can}
If $\Gamma$ is a context containing no variable assignments, $\Gamma\vdash t:A$, and $t$ is normal, then the main term-former of $t$ is the introduction for the
main type-former of $A$.
\end{theorem}

Concretely, if $A$ is some base type then $t$ is a value of that type.

\begin{theorem}[Subformula Property]\label{thm:subform}
Given $\Gamma\vdash t:A$ with $t$ normal, all subterms of $t$ have as their type in
the derivation tree a subtype of $A$, or a subtype of a type assigned in $\Gamma$.
\end{theorem}

To attain this final theorem we need to take some care with sums. It is well known
that lambda calculi with sums do not enjoy the subformula property unless they have
additional reductions called commuting conversions~\cite[Chapter 10]{Girard:Proofs}.
However the commuting conversions for the $\Box$ type
\begin{align*}
  \open\,\case{s}{x}{t}{y}{u} &\mapsto \case{s}{x}{\open\,t}{y}{\open\,u} \\
  \open\,\abort\,t &\mapsto \abort\,t
\end{align*}
do not obviously enjoy subject reduction because $\open$ might change the context.
However if we tweak the definitions of the elimination term-formers for sums
according to Figure~\ref{fig:sums} then all results of this section indeed hold.

\begin{figure}[t]
  \begin{mathpar}
  \inferrule*{
      \Gamma \vdash s: A+B \\
      \Gamma,x:A,\Gamma' \vdash t : C \\
      \Gamma,y:B,\Gamma' \vdash u : C}{
      \Gamma,\Gamma' \vdash \case{s}{x}{t}{y}{u} : C}
    \and
    \inferrule*{
      \Gamma \vdash t:0}{
      \Gamma,\Gamma' \vdash \abort\,t : A}
  \end{mathpar}
  \caption{Elimination term-formers for sums}
  \label{fig:sums}
\end{figure}

Finally, while we will not explore computational aspects of $\eta$-equivalence in this
paper, we do note that
\[
  \shut\,\open\,t \;=\; t
\]
obeys subject reduction in both directions (provided, in the expansion case, that the
type of $t$ has $\Box$ as its main type-former).

\subsection{Categorical Semantics}\label{subsec:IKcat}

This section goes beyond Theorem~\ref{thm:IKlogsc} to establish the soundness of
the type system with respect to a \emph{categorical semantics}, in cartesian closed
categories $\Cat$ equipped with an endofunctor $\Box$ that has a \emph{left
adjoint}, which we write $\Dia$.

We interpret types as $\Cat$-objects via the structure of $\Cat$ in the obvious
way. We then interpret contexts as $\Cat$-objects by
\begin{itemize}
  \item $\den{\cdot} \;\defeq\; 1$;
  \item $\den{\Gamma,x:A} \;\defeq\; \den{\Gamma}\times A$;
  \item $\den{\Gamma,\lock} \;\defeq\; \Dia\den{\Gamma}$.
\end{itemize}
We omit the brackets $\den{\cdots}$ where no confusion is possible, and usually
abuse notation by omitting the left-most `$1\times$' where the left of the context is a
variable.

We will also sometimes interpret contexts $\Gamma$ as endofunctors,
abusing notation to also write them as $\den{\Gamma}$, or merely $\Gamma$, by
taking $\den{\cdot}$ as the identity, $\den{\Gamma,x:A} = \den{\Gamma}\times A$,
and $\den{\Gamma,\lock} = \Dia\den{\Gamma}$.

We interpret $\Gamma\vdash t:A$ as a $\Cat$-arrow
$\den{\Gamma\vdash t:A}:\den{\Gamma}\to A$, often abbreviated to $\den{t}$, or
merely $t$, by induction on the derivation of $t$ as follows.

Standard constructions such as variables, abstraction and application are interpreted as
usual. To interpret the rules for sums of Figure~\ref{fig:sums} we use the fact that
$\Dia$, as a left adjoint, preserves colimits.

$\shut$: we simply apply the isomorphism $\Cat(\Dia\den{\Gamma},A)\to
\Cat(\den{\Gamma},\Box A)$ given by the $\Dia\dashv\Box$ adjunction.

$\open$: We apply the isomorphism $\Cat(\den{\Gamma},\Box A)\to
\Cat(\Dia\den{\Gamma},A)$ to the arrow interpreting the premise, then compose with
the projection $\den{\Gamma,\lock,\Gamma'}\to\den{\Gamma,\lock}$.

\begin{theorem}[Categorical Soundness]\label{thm:IKcatsnd}
If $\Gamma\vdash t:A$ and $t\mapsto t'$ then $\den{t}=\den{t'}$.
\end{theorem}

We also have that $\eta$-equivalent terms have the same denotation.


\section{Left Adjoints and Categorical Completeness}\label{sec:left}

In this section we extend the calculus to include the left adjoint $\Dia$ as a first-class
type-former, and hence prove categorical completeness. The underlying logic is the
fragment of intuitionistic tense logic~\cite{Ewald:Intuitionistic} with just one pair of
modalities, studied by Dzik et al.~\cite{Dzik:Intuitionistic} as `intuitionistic
logic with a Galois connection'; we use the name $\IKd$. We have two new axioms

  $\eta^m$: $\;A\to\Box\Dia A$

  $\varepsilon^m$: $\;\Dia\Box A\to A$ \\
We use the superscript $m$ to identify these as the unit as the unit and counit of the
\emph{modal} adjunction $\Dia\dashv\Box$, to differentiate them from other
(co)units used elsewhere in the paper. We have one new inference rule:

  \emph{Monotonicity: } if $A\to B$ is a theorem, then so is $\Dia A\to\Dia B$.

\subsection{Type System and Computation}
\label{sec:leftadjbasic}

We extend the type system of Figure~\ref{fig:IKtyping} with the new rules
for $\Dia$ presented in Figure~\ref{fig:adjtyping}. $\Dia$, unlike $\Box$, need not
commute with products, so does not interact well with contexts. Hence the
subterms of a $\letdbare$ term may not share variables.

\begin{figure}[t]
  \begin{mathpar}
    \inferrule*[right=$\lock\notin\Gamma'$]{
      \Gamma \vdash t: A}{
      \Gamma,\lock,\Gamma' \vdash \dia\,t: \Dia A}
    \and
    \inferrule*{
      \Gamma\vdash t:\Dia A \\
      x:A,\lock\vdash u:B}{
      \Gamma\vdash \letd{x}{t}{u}:B}
  \end{mathpar}
  \caption{Additional typing rules for logic $\IKd$}
  \label{fig:adjtyping}
\end{figure}

We can construct the axioms of $\IKd$:
\begin{mathpar}
\inferrule*{
  x:A,\lock\vdash\dia\,x:\Dia A }{
  x:A\vdash\shut\,\dia\,x:\Box\Dia A}
\and
\inferrule*{
   x:\Dia\Box A\vdash x:\Dia\Box A \\
   y:\Box A,\lock\vdash\open\,y:A}{
  x:\Dia\Box A\vdash\letd{y}{x}{\open\,y}:A}
\end{mathpar}
and given a closed term $f:A\to B$ we have the monotonicity construction
\[
\inferrule*{
  x:\Dia A\vdash x:\Dia A \\
  y:A,\lock\vdash\dia(f\,y):\Dia B}{
  x:\Dia A\vdash\letd{y}{x}{\dia(f\,y)}:\Dia B}
\]

To this we add the new $\beta$ rule
\[
  \letd{x}{\dia\,t}{u} \;\mapsto\; u[t/x]
\]

We can hence extend the syntactic results of the previous section to the logic $\IKd$,
with the exception of the subformula property. Consider the term
\[
\inferrule*{
  x:\Dia A\vdash \letd{y}{x}{\lambda z.\dia\,y}:\Dia A\to\Dia A \\
  x:\Dia A\vdash x:\Dia A}{
  x:\Dia A\vdash (\letd{y}{x}{\lambda z.\dia\,y})x:\Dia A}
\]
This term is normal but evidently fails the subformula property. One might expect, as
with sums, that a commuting conversion would save the day by reducing the term to
$\letd{y}{x}{((\lambda z.\dia\,y)x)}$, but this term sees the free variable $x$ appear
in the second subterm of a $\letdbare$ expression, which is not permitted.

We now turn to $\eta$-equivalence, and an equivalence which
we call \emph{associativity}:
\begin{align*}
  \letd{x}{t}{\dia\,x} &\;=\; t \\
  \letd{x}{s}{(t[u/y])} &\;=\; t[\letd{x}{s}{u}/y] \mbox{ if $t$'s context contains $y$ only}
\end{align*}
For example, under associativity the counter-example to the subformula property
equals $(\lambda z.\letd{y}{x}{\dia\,y})x$, which reduces to $\letd{y}{x}{\dia\,y}$,
which is $\eta$-equal to $x$. The equivalences enjoy subject reduction in both
directions (requiring, as usual, that $t$ has the right type for $\eta$-expansion). 

\subsection{Categorical Semantics}

We interpret the new term-formers in the same categories as used in
Section~\ref{subsec:IKcat}. For $\dia$, given $t:\Gamma\to A$ we compose $\Dia t$
with the projection $\Gamma,\lock,\Gamma'\to\Gamma,\lock$. The denotation of
$\letd{x}{t}{u}$ is simply $u\circ t$. We may then confirm the soundness of
$\beta$-reduction, $\eta$-equivalence, and associativity; we call these equivalences
collectively \emph{definitional equivalence}.

We extend standard techniques for proving
completeness~\cite{Lambek:Introduction}, constructing a \emph{term model},
a category with types as objects and, as arrows $A\to B$, terms of form
$x:A\vdash t:B$ modulo definitional equivalence. This is a category by taking
identity as the term $x$ and composition $u\circ t$ as $u[t/x]$. It is a cartesian
closed category using the type- and term-formers for products and function spaces.

The modalities $\Dia$ and $\Box$ act on types; they also act on terms by, for $\Dia$,
the monotonicity construction, and for $\Box$, mapping $x:A
\vdash t:B$ to $x:\Box A\vdash\shut\,t[\open\,x/x]:\Box B$. One can check these
constructions are functorial, and that the terms for $\eta^m$ and $\varepsilon^m$
are natural and obey the triangle equalities for the adjunction $\Dia\vdash\Box$.

Given a context $\Gamma$ we define the \emph{context term} $\Gamma\vdash
c_{\Gamma}:\den{\Gamma}$ by
\begin{itemize}
\item $c_{\cdot} \;\defeq\; \langle\rangle$;
\item $c_{\Gamma,x:A} \;\defeq\; \langle c_{\Gamma},x\rangle$;
\item $c_{\Gamma,\sslock} \;\defeq\; \dia\,c_{\Gamma}$.
\end{itemize}

\begin{lemma}\label{lem:compl_lem}
Given $\Gamma\vdash t:A$, $t$ is definitionally equal to $\den{\Gamma\vdash t:A}
[c_{\Gamma}/x]$.
\end{lemma}

\begin{theorem}[Categorical Completeness]
If $\Gamma\vdash t:A$ and $\Gamma\vdash u:A$ are equal in all models then they
are definitionally equal.
\end{theorem}
\begin{proof}
$t$ and $u$ have equal denotations in the term model, so their denotations are
definitionally equal. Definitional equality is preserved by substitution, so
$\den{\Gamma\vdash t:A}[c_{\Gamma}/x]=
\den{\Gamma\vdash u:A}[c_{\Gamma}/x]$, so by Lemma~\ref{lem:compl_lem},
$t=u$.
\end{proof}


\section{Intuitionistic S4 for Idempotent Comonads}\label{sec:S4}

Intuitionistic S4 (IS4) is the extension of IK with the axioms

  T: $\;\Box A\to A$

  4: $\;\Box A\to \Box\Box A$ \\
To the category theorist IS4 naturally suggests the notion of a
\emph{comonad}. IS4 is one of the most studied and widely applied intuitionistic
modal logics; in particular there exist two Fitch-style
calculi~\cite{Pfenning:Modal,Davies:Modal}. We conjecture that similar results to the
previous sections could be developed for these calculi. Instead of pursuing such a result,
we here show that a simpler calculus is possible if we restrict to
\emph{idempotent} comonads, where $\Box A$ and $\Box\Box A$ are isomorphic.
This restriction picks out an important class of examples --  see for example the
discussion of Rijke et al.~\cite{Rijke:Modalities} -- and relies on a novel
`coherence' proof.

\subsection{Type System and Computation}\label{Sec:S4type}

A calculus for IS4 is obtained by replacing the $\open$ rule of
Figure~\ref{fig:IKtyping} by
\[
\inferrule*{
  \Gamma\vdash t:\Box A}{
  \Gamma,\Gamma'\vdash \open\,t:A
}
\]

The T and 4 axioms are obtained by
\begin{mathpar}
  \inferrule*{
    x:\Box A\vdash x:\Box A}{
    x:\Box A\vdash \open\,x:A}
  \and
  \inferrule*{
    \inferrule*{
        x:\Box A,\lock,\lock\vdash\open\,x:A}{
      x:\Box A,\lock\vdash\shut\,\open\,x:\Box A}}{
    x:\Box A\vdash\shut\,\shut\,\open\,x:\Box\Box A}
\end{mathpar}

This confirms logical completeness; once can also easily check soundness.

Subject reduction for the $\beta$-reduction $\open\,\shut\,t\mapsto t$
requires a new lemma, proved by an easy induction on $t$:

\begin{lemma}[Lock Replacement]\label{lem:lockrep}
If $\Gamma,\lock,\Gamma''\vdash t:A$ then $\Gamma,\Gamma',\Gamma''\vdash t:A$.
\end{lemma}

The key syntactic theorems~\ref{thm:sn},~\ref{thm:can}, and~\ref{thm:subform}
then follow easily.

$\eta$-expansion obeys subject reduction as before, but it is not the case, for
example, that the term presented
above for the 4 axiom reduces to $\shut\,x$. We may however accept a notion of
$\eta$-reduction on typed terms-in-context:
\[
  \Gamma\vdash\shut\,\open\,t\mapsto t:\Box A \mbox{ provided that }
  \Gamma\vdash t:\Box A
\]
This equivalence is more powerful than it might appear; it allows us to derive the
idempotence of $\Box$, as the 4 axiom is mutually inverse with the instance
$\Box\Box A\to\Box A$ of the T axiom. That is,
$\lambda x.\open\,\shut\,\shut\,\open\,x$ reduces to the identity
on $\Box A$, and $\lambda x.\shut\,\shut\,\open\,\open\,x$ reduces to the identity on
$\Box\Box A$.

\subsection{Categorical Semantics}

We give semantics to our type theory in a cartesian closed category with an
adjunction of endofunctors $\Dia\dashv\Box$ in which $\Box$ is a \emph{comonad}.
Equivalently~\cite[Section 3]{Eilenberg:Adjoint}, $\Dia$ is a monad, equipped with
a unit $\eta$ and multiplication $\mu$.
To confirm the \emph{coherence} of these semantics,
discussed in the next subsection, and the soundness of $\eta$-equivalence, we further
require that $\Box$ is idempotent, or equivalently that all
$\mu_A:\Dia\Dia A\to\Dia A$ are isomorphisms with inverses
$\eta_{\Dia A}=\Dia\eta_A$.

To define the semantics we define \emph{lock replacement} natural transformations
$l_{\Gamma}:\den{\Gamma}\to\Dia$, corresponding to Lemma~\ref{lem:lockrep},
by induction on $\Gamma$:
\begin{itemize}
\item $l_{\cdot}$ is the unit $\eta$ of the monad;
\item $l_{\Gamma,x:A}$ is the projection composed with $l_{\Gamma}$;
\item $l_{\Gamma,\sslock}$ is $\Dia l_{\Gamma}$ composed with $\mu$.
\end{itemize}
Note that $l_{\sslock}$ is the identity by the monad laws.

We may now define the interpretation of $\open$: given $t:\Gamma\to\Box A$
we apply the adjunction to get an arrow $\Dia\Gamma\to A$, then compose with
$l_{\Gamma'}:\Gamma,\Gamma'\to\Gamma,\lock$.

\begin{lemma}\label{lem:lock_over}
If we replace part of a context with a lock, then replace part of the new
context that includes the new lock, we could have done this in one step:
\[\xymatrix{
  \Gamma_1,\Gamma_2,\Gamma_3,\Gamma_4 \ar[rr]^-{l_{\Gamma_2,\Gamma_3,\Gamma_4}} \ar[dr]_-{\Gamma_4(l_{\Gamma_3})} & & \Gamma_1,\lock \\
  & \Gamma_1,\Gamma_2,\lock,\Gamma_4 \ar[ur]_-{\,l_{\Gamma_2,\tlock,\Gamma_4}}
}\]
\end{lemma}
\begin{proof}
By induction on $\Gamma_4$, with the base case
following by induction on $\Gamma_3$.
\end{proof}


\begin{lemma}\label{lem:lockreponterms}
$\den{\Gamma,\lock,\Gamma''\vdash t:A}\circ \den{\Gamma''}(l_{\Gamma'})=
\den{\Gamma,\Gamma',\Gamma''\vdash t:A}$.
\end{lemma}
\begin{proof}
By induction on the derivation of $t$.
\end{proof}

Now $\open\,\shut\,t$, where the $\open$ has weakening $\Gamma'$, has denotation
$\varepsilon^{m}\circ\Dia\Box t\circ\Dia\eta^{m}\circ l_{\Gamma'}$, which
is $t\circ l_{\Gamma'}$ by the naturality of $\varepsilon^{m}$, and the
adjunction. This is what is required by Lemma~\ref{lem:lockreponterms}, so
$\beta$-reduction for $\Box$ is soundly modelled.

\subsection{Coherence}

Because the $\open$ rule involves a weakening, and does not explicitly record in the
term what that weakening is, the same typed term-in-context can be the root of
multiple derivation trees, for example:
\begin{mathpar}
  \inferrule*{
    \inferrule*{
      x:\Box\Box A\vdash x:\Box\Box A}{
      x:\Box\Box A,\lock\vdash \open\,x:\Box A}}{
    x:\Box\Box A,\lock,\lock\vdash\open\,\open\,x:A}
  \and
  \inferrule*{
    \inferrule*{
      x:\Box\Box A\vdash x:\Box\Box A}{
      x:\Box\Box A\vdash \open\,x:\Box A}}{
    x:\Box\Box A,\lock,\lock\vdash\open\,\open\,x:A}
\end{mathpar}
The categorical semantics of the previous section is defined by induction on
derivations, and so does not truly give semantics to \emph{terms} unless any two
trees with the same root must have the same denotation. In this section we show that
this property, here called \emph{coherence}, indeed holds. We make crucial use of
the idempotence of the comonad $\Box$.

We first observe that if $\Gamma,\Gamma',\Gamma''\vdash t:A$ and all variables of
$\Gamma'$ are not free in $t$, then $\Gamma,\Gamma''\vdash t:A$. The following
lemma, proved by easy inductions, describes how the denotations of these derivations
are related:

\begin{lemma}\label{lem:strength}
\begin{enumerate}
\item
If $x$ is not free in $t$ then $\Gamma,x:A,\Gamma'\vdash t:B$ has the same
denotation as $\Gamma,\Gamma'\vdash t:B\circ\Gamma'(pr)$.
\item
$\Gamma,\Gamma'\vdash t:B$ has denotation $\Gamma,\lock,\Gamma'\vdash t:B
\circ\Gamma'(\eta)$.
\end{enumerate}
\end{lemma}

The technical lemma below is the only place where
idempotence is used.

\begin{lemma}\label{lem:strength_and_l}
Given $\Gamma,\Gamma'\vdash t:A$ with $\Gamma'$ not free in $t$, we have
\[\xymatrix{
  \Gamma,\Gamma' \ar[r]^-{t} \ar[d]_{l_{\Gamma'}} & A \ar[d]^{\eta} \\
  \Gamma,\lock \ar[r]_-{\Dia t} & \Dia A
}\]
where $t$ on the bottom line is the original arrow with $\Gamma'$ strengthened
away.
\end{lemma}
\begin{proof}
By induction on $\Gamma'$. The base case holds by the naturality of $\eta$.

We present only the lock case: $\eta\circ t=\Dia t\circ \eta$ by the naturality of
$\eta$. But by
\textbf{idempotence}, $\eta:\Gamma,\Gamma',\lock\to\Gamma,\Gamma',\lock,\lock$
equals
$\Dia\eta$. Then by Lemma~\ref{lem:strength} $\Dia t\circ\Dia\eta$ is
$\Dia\den{\Gamma,\Gamma'\vdash t:A}$, i.e. we have strengthened the lock away
and can hence use our induction hypothesis, making the top trapezium commute in:
\[\xymatrix{
  \Gamma,\Gamma',\lock \ar[rrr]^(0.57){\Dia t} \ar[dr]_-{\Dia l_{\Gamma'}} \ar[dd]_{l_{\Gamma',\tlock}} & & & \Dia A \ar[dl]^-{\Dia\eta} \ar[dd]^{id} \\
  & \Gamma,\lock,\lock \ar[r]^-{\Dia\Dia t} \ar[dl]_-{\mu} & \Dia\Dia A \ar[dr]^-{\mu} \\
  \Dia{\Gamma} \ar[rrr]_(0.57){\Dia t} & & & \Dia A
}\]
The left triangle commutes by definition, the bottom trapezium commutes by the
naturality of $\mu$, and the right triangle commutes by the monad laws.
\end{proof}

\begin{lemma}\label{lem:strength_and_l2}
Given $\Gamma,\Gamma'\vdash t:A$ with $\Gamma'$ not free in $t$, we have
\[\xymatrix{
  \Gamma,\Gamma',\Gamma'' \ar[r]^-{l_{\Gamma''}} \ar[d]_{l_{\Gamma',\Gamma''}} & \Gamma,\Gamma',\lock \ar[d]^{\Dia t} \\
  \Gamma,\lock \ar[r]_-{\Dia t} & \Dia A
}\]
where the bottom $t$ is obtained via strengthening.
\end{lemma}
\begin{proof}
By induction on $\Gamma''$. The base case follows by
Lemma~\ref{lem:strength_and_l}.
\end{proof}

\begin{lemma}\label{lem:coherence}
Given $\Gamma,\Gamma'\vdash t:\Box\,A$ with the variables of $\Gamma'$ not free in
$t$, the following arrows are equal:
\begin{itemize}
\item $\Gamma,\Gamma',\Gamma''\vdash\open\,t:A$ where the weakening is
$\Gamma''$;
\item obtaining an arrow $\Gamma\to\Box A$ via Lemma~\ref{lem:strength}, then
applying $\open$ with weakening $\Gamma,\Gamma''$.
\end{itemize}
\end{lemma}
\begin{proof}
Immediate from Lemma~\ref{lem:strength_and_l2}, i.e.
\[\xymatrix{
  \Gamma,\Gamma',\Gamma'' \ar[r]^-{l_{\Gamma''}} \ar[d]_{l_{\Gamma',\Gamma''}} & \Gamma,\Gamma',\lock \ar[d]^{\Dia t} \\
  \Gamma,\lock \ar[r]_-{\Dia t} & \Dia \Box A \ar[r]_-{\varepsilon^{m}} & A
}\]
\end{proof}

\begin{theorem}[Coherence]\label{thm:cohere}
Given two different derivation trees of a term, their denotation is equal.
\end{theorem}
\begin{proof}
By induction on the number of nodes in the trees. The base case
with one node is trivial.
Suppose we have $n+1$ nodes. Then the induction hypothesis immediately completes
the proof unless the nodes above the roots are non-equal. Then the final construction
must be an instance of $\open$, i.e. we have
\begin{mathpar}
  \inferrule*{
    \Gamma\vdash t:\Box\,A}{
    \Gamma,\Gamma',\Gamma''\vdash \open\,t:A}
  \and
  \inferrule*{
      \Gamma,\Gamma'\vdash t:\Box\,A}{
      \Gamma,\Gamma',\Gamma''\vdash \open\,t:A}
\end{mathpar}
Clearly any variables in $\Gamma'$ are not free in $t$, so we can use
Lemma~\ref{lem:strength} on the top line of the right hand tree to derive
$\Gamma\vdash t:\Box\,A$. By induction hypothesis this has the same
denotation as the top line of the left hand tree. But Lemma~\ref{lem:coherence}
tells us that applying this strengthening and then opening with $\Gamma',\Gamma''$ is
the same as opening with $\Gamma''$ only.
\end{proof}

We can now demonstrate the soundness of $\eta$-equivalence:
given $\Gamma\vdash t:\Box A$ and $\Gamma\vdash\shut\,\open\,t:\Box A$ by any
derivations, we can by coherence safely assume that $\open$ used one lock only as
its weakening, and so the arrows are equal by the $\Dia\dashv\Box$ adjunction.

\subsection{Left Adjoints and Categorical Completeness}

Following Section~\ref{sec:left} we can add $\Dia$ to the type theory; we need only
modify the $\dia$ rule to
\[
    \inferrule*{
    \Gamma \vdash t: A}{
    \Gamma,\Gamma' \vdash \dia\,t: \Dia A}
\]
to retain Lemma~\ref{lem:lockrep}. The results of the previous sections, apart once
more for the subformula property, still hold, where we define the
denotation of $\Gamma,\Gamma'\vdash\dia\,t$ as $\Dia t$ composed
with $l_{\Gamma'}$. In particular, we must confirm that
Lemma~\ref{lem:compl_lem} extends to the new definitions of $\open$ and $\dia$,
for which we need the lemma below:

\begin{lemma}\label{lem:lockrep_and_ctxterm}
Given the term $x:\den{\Gamma,\Gamma'}\vdash l_{\Gamma'}:\Dia\den{\Gamma}$
defined in the term model, $l_{\Gamma'}[c_{\Gamma,\Gamma'}/x]$ is definitionally
equal to $\dia\,c_{\Gamma}$.
\end{lemma}

Now $\den{\open\,t}[c_{\Gamma,\Gamma'}/x]$ is $\letd{x}{(\letd{x}{l_{\Gamma'}
[c_{\Gamma,\Gamma'}/x]}{\dia\den{t}})}{\open\,x}$, which by the lemma above is
$\letd{x}{(\letd{x}{\dia\,c_{\Gamma}}{\dia\den{t}})}{\open\,x}\mapsto$
$\open\den{t}[c_{\Gamma}/x]$, which equals $\open\,t$ by induction. The proof for
$\dia$ is similar.


\section{Intuitionistic R}\label{sec:IR}

One can readily imagine how the calculus for IS4 could be modified for logics with
only one of the T and 4 axioms. In this section we instead illustrate the flexibility of
Fitch-style calculi by defining a calculus for the rather
different logic Intuitionistic R (IR), which extends IK with the axiom

  R: $\;A\to\Box A$ \\
This axiom was first studied for
intuitionistic necessity modalities by Curry~\cite{Curry:Theory},
along with the axiom M, $\Box\Box A\to\Box A$, to develop a logic for monads. The
importance of the logic with R but without M was established by McBride and
Paterson~\cite{McBride:Applicative} who showed that it captured the useful
programming abstraction of \emph{applicative functors}. We take the name R
for the axiom from Fairtlough and Mendler~\cite{Fairtlough:Propositional}, and for
the logic from Litak~\cite{Litak:Constructive}.

We modify Figures~\ref{fig:IKtyping} and~\ref{fig:adjtyping} simply by
removing the side-conditions $\lock\notin\Gamma$ from the variable, $\open$, and
$\dia$ rules. We can then derive R:
\[
  \inferrule*{
    x:A,\lock\vdash x:A}{
    x:A\vdash \shut\,x:\Box A}
\]

For substitution and subject reduction we require the following lemma, easily proved
by induction on the derivation of $t$:

\begin{lemma}[Lock Weakening]\label{lem:tickweak}
If $\Gamma,\Gamma'\vdash t:A$ then $\Gamma,\lock,\Gamma'\vdash t:A$.
\end{lemma}

We can also observe that $\eta$-equivalence preserves types in both directions.

We give semantics for this calculus in a cartesian closed category equipped with an
adjunction of endofunctors $\Dia\dashv\Box$ and a `point' natural transformation
$r:Id\to\Box$ preserved by $\Box$, i.e. $\Box r=r:\Box A\to\Box\Box A$. This last
property makes this model slightly less general than the notion of tensorial strength
used for categorical semantics by McBride and
Paterson~\cite{McBride:Applicative}, but is needed for coherence and the
soundness of $\eta$-equivalence. We will use the arrow $\Dia A\to A$ defined
by applying the adjunction to $r$; we call this $q$ and note the property:

\begin{lemma}\label{q_preserved}
$q=\Dia q:\Dia\Dia A\to\Dia A$.
\end{lemma}

The \emph{weakening} natural transformation $w_{\Gamma}:\Gamma\to Id$ is
defined by induction on $\Gamma$ via projection and $q$. Variables are then
denoted by projection composed with weakening, and weakening is
used similarly for $\open$ and $\dia$. We can hence show the soundness
of $\beta$-reduction for $\Box$ and $\Dia$.
For the soundness of $\eta$-equivalence for $\Box$ we need the following lemma:

\begin{lemma}\label{lem:weak_and_tick}
$w_{\Gamma',\sslock}=\Dia w_{\sslock,\Gamma'}:\Gamma,\lock,\Gamma',\lock\to
\Gamma,\lock$.
\end{lemma}

The
denotation of $\Gamma,\lock,\Gamma'\vdash\shut\,\open\,t$ is $\Box\varepsilon^m
\circ\Box\Dia t\circ\Box w_{\Gamma',\sslock}\circ\eta^m$. By the above lemma
we replace $\Box w_{\Gamma',\sslock}$ with $\Box\Dia w_{\sslock,\Gamma'}$, so
by the naturality of $\eta^m$ we have $\Box\varepsilon^m\circ\eta^m\circ t\circ
w_{\sslock,\Gamma'}$, which is $t\circ w_{\sslock,\Gamma'}$ by the monad laws.

Moving to coherence, we conduct a similar induction to Theorem~\ref{thm:cohere},
considering the case
\begin{mathpar}
  \inferrule*{
    \Gamma\vdash t:\Box\,A}{
    \Gamma,\lock,\Gamma',\lock,\Gamma''\vdash \open\,t:A}
  \and
  \inferrule*{
      \Gamma,\lock,\Gamma'\vdash t:\Box\,A}{
      \Gamma,\lock,\Gamma',\lock,\Gamma''\vdash \open\,t:A}
\end{mathpar}
The top line on the left weakens to the top line on the right, with denotation
$t\circ w_{\sslock,\Gamma'}$. By induction this equals the denotation of the top line
of the right. Then the right hand term has denotation $\varepsilon^m\circ\Dia t
\circ\Dia w_{\sslock,\Gamma'}\circ w_{\Gamma''}$. But by
Lemma~\ref{lem:weak_and_tick} $\Dia w_{\sslock,\Gamma'}=
w_{\Gamma',\sslock}$.
It is clear that $w_{\Gamma',\sslock}\circ w_{\Gamma''}=
w_{\Gamma',\sslock,\Gamma''}$, which is exactly the weakening used on the left.
Coherence for $\dia$ follows similarly.

Moving finally to categorical completeness, in the term model $\Box t\circ r$ is
$\shut\,t[\open\,\shut\,x/x]$, which reduces to $\shut\,t$, so $r$ is natural.
$\Box r:\Box A\to\Box\Box A$ is $\shut\,\shut\,\open\,x$, which is indeed
$\eta$-equal to $\shut\,x$.

We finally need to update Lemma~\ref{lem:compl_lem} for our new definitions. We
do this via a lemma similar to Lemma~\ref{lem:lockrep_and_ctxterm}:

\begin{lemma}\label{lem:weak_and_ctx}
Given the term $x:\den{\Gamma,\Gamma'}\vdash w_{\Gamma'}:\den{\Gamma}$
defined in the term model, $w_{\Gamma'}[c_{\Gamma,\Gamma'}/x]$ is definitionally
equal to $c_{\Gamma}$.
\end{lemma}

Now the denotation of $\Gamma,x:A,\Gamma'\vdash x:A$ is $\pi_2 w_{\Gamma'}$.
Therefore we have $\pi_2 w_{\Gamma'}[c_{\Gamma,A,\Gamma'}/x]$, which is
$\pi_2 c_{\Gamma,A}$ by the lemma above. This is $\pi_2\langle c_{\Gamma},x
\rangle$, which reduces to $x$.

The denotation of $\Gamma,\lock,\Gamma'\vdash\open\,t:A$ is
$\letd{x}{w_{\Gamma'}}{\open\den{t}}$. Applying the substitution
$[c_{\Gamma,\sslock,\Gamma'}/x]$ along with the lemma above yields the term
$\letd{x}{\dia\,c_{\Gamma}}{\open\den{t}}\mapsto\open\den{t}[c_{\Gamma}/x]$,
and induction completes. The calculations for $\dia$ follow similarly.


\section{Related and Further Work}\label{sec:concl}

\textbf{\emph{Conventional contexts.}} Lambda calculi with conventional contexts
containing typed variables only have been proposed for the logic of
monads~\cite{Moggi:Computational}, for IS4~\cite{Bierman:Intuitionistic}, for
IK~\cite{Bellin:Extended}, and for a logic with `L\"{o}b
induction'~\cite{Birkedal:Intensional}, from which one can extract a calculus
for IR. In previous work~\cite{Clouston:Guarded} we developed the \emph{guarded
lambda calculus} featuring two modalities, where one (`constant') was an
(idempotent) comonad,
and the other (`later') supported a notion of guarded recursion corresponding to
L\"{o}b induction. We therefore used the existing
work~\cite{Bierman:Intuitionistic,Birkedal:Intensional} `off the shelf'.

Problems arose when we attempted to extend our calculus with
dependent types~\cite{Bizjak:Guarded}. Neither of the calculi with conventional
contexts we had used scaled
well to this extension. The calculus for IS4~\cite{Bierman:Intuitionistic}, whose terms
involved explicit substitutions, turned out to require these substitutions on types also,
which added a level of complexity that made it difficult to write even quite
basic dependently typed programs. The constant modality was
therefore jettisoned in favour of an approach based on clock
quantification~\cite{Atkey:Productive}, of which more below. The calculus for
later employed a connective $\circledast$ (from McBride and
Patterson~\cite{McBride:Applicative}) which acted on
function spaces under the
modality. However with dependent types we need to act not merely on function
spaces, but on $\Pi$-types, and $\circledast$ was unable to be used. Instead a novel
notion of
`delayed substitution' was introduced. These were given an equational theory, but
some of
these equations could not be directed, so they did not give rise to a useful notion of
computation.

\textbf{\emph{Modalities as quantifiers.}} The suggestive but formally rather
underdeveloped paper of De Queiroz and Gabbay~\cite{deQueiroz:Functional}
proposed that necessity modalities should be treated as universal quantifiers, inspired
by the standard semantics of necessity as
`for all possible worlds'. This is one way to understand the relationship
between the constant modality and clock quantification~\cite{Atkey:Productive}.
However clock quantification is more general than a single
constant modality because we can identify multiple free clock variables with multiple
`dimensions' in which a type may or may not be constant. This gap in generality can
probably be bridged by using multiple independent constant modalities. More
problematically, while it is clear what the denotational semantics of the constant
modality are, the best model for clock quantifiers yet
found~\cite{Bizjak:Denotational} is
rather complicated and still leaves open some problems with coherence in the
presence of a universe.

\textbf{\emph{Previous Fitch-style calculi.}} The Fitch-style approach was pioneered,
apparently independently, by Martini and
Masini~\cite{Martini:Computational} and Borghuis~\cite{Borghuis:Coming}. Martini
and Masini's work is rather notationally heavy, and weakening appears not to be
admissible. Borghuis's calculus for IK is excellent, but his calculi for stronger logics
are not so compelling, as each different axiom is expressed with another version of
the $\open$ or $\shut$ rules, not all of which compute when combined. The
calculus for IS4 of Pfenning and Wong~\cite{Pfenning:Modal}, refined by
Davies and Pfenning~\cite[Section 4]{Davies:Modal}, provide the basis of
the IS4 calculus of this paper, but involve some complications which appear to
correlate to not assuming idempotence. We have extended this previous work by
investigating the subformula property, introducing categorical semantics, and showing
how left adjoints to necessity modalities \`{a} la tense logic can be used as types.
Finally, the recent clocked type theory of Bahr
et al.~\cite{Bahr:Clocks} independently gave a treatment of the later modality that on
inspection is precisely Fitch-style (albeit with named `locks'), and which has
better computational properties than the delayed substitution approach.

\textbf{\emph{Dual contexts.}} Davies and
Pfenning~\cite{Davies:Modal} use a pair of contexts $\Delta;
\Gamma$ with intended meaning $\Box\Delta\land\Gamma$. This is quite different
from the semantics of Fitch-style sequents, where structure in the context
denotes the \emph{left adjoint} of $\Box$. In recent work
Kavvos~\cite{Kavvos:Dual} has shown that dual contexts may capture a number of
different modal logics, and the approach has been used as a foundation for both
pen-and-paper mathematics~\cite{Shulman:Brouwer} and, via an Agda
fork~\cite{Vezzosi:Agda}, formalisation~\cite{Licata:Internal}.
We support this work but there is reason to explore other
options. First, writing programs with dual context calculi
was described by Davies and Pfenning themselves as
`somewhat awkward', and in the same paper they suggest the Fitch-style
approach as a less awkward alternative. Indeed, Fitch's approach was exactly designed to
capture `natural' modal deduction.
Second, any application with multiple interacting modalities is unlikely to
be accommodated in a mere two zones; the \emph{mode theories} of Licata
et al.~\cite{Licata:Fibrational} extend the dual zone approach to a richer setting in which
interacting modalities, substructural contexts, and even Fitch-style natural deduction can
be expressed%
\footnote{We are grateful to an anoymous reviewer for this last observation.}%
, but the increase in complexity is considerable and much work remains to be done.

\textbf{\emph{Further logics and algorithmic properties.}} We wish to bring more
logics into the Fitch-style framework, in particular the logic of the later modality,
extending IR with the strong L\"{o}b axiom
$(\Box A\to A)\to A$. The obvious treatment of this axiom does not terminate.
but Bahr et al.~\cite{Bahr:Clocks} suggest that this can be managed by giving
names to locks. We would further like to develop calculi
with multiple modalities. This is easy to do by assigning each modality
its own lock; two IK modalities give exactly the intuitionistic tense logic of
Gor{\'e} et al.~\cite{Gore:Cut}. The situation is rather more interesting where
the modalities interact, as with the later and constant modalities. Finally, we would like to
further investigate algorithmic properties of Fitch-style calculi such as type checking, type
inference, and $\eta$-expansion and other notions of computation. In particular, we
wonder if a notion of commuting conversion can be defined so that the calculi with
$\Dia$ enjoy the subformula property.


\newpage
\bibliographystyle{splncs03}
\bibliography{bibl}

\newpage
\appendix


\section{Intuitionistic K}

This appendix presents proof details for the theorems of Section~\ref{sec:IK}. We
omit routine proof details for products, and sometimes function spaces also, and we
delay discussion of sums until Appendix~\ref{app:sums}.

\subsection{Proof of Theorem~\ref{thm:IKlogsc} (Logical Soundness and Completeness)}

In this section we prove the soundness and completeness of the type system of
Figure~\ref{fig:IKtyping} (considered as a natural deduction system) with respect to
the logic $IK$. The typing rules for the connectives of intuitionistic logic are as usual
and so soundness and completeness for this fragment is clear.

For logical completeness we then need only show that the $K$ axiom is derivable,
which is done in Section~\ref{Sec:IKtype}, and that necessitation holds. For this
we need the lemma:

\begin{lemma}[Left Weakening]\label{lem:left_weak}
If $\Gamma'\vdash t:A$ then $\Gamma,\Gamma'\vdash t:A$.
\end{lemma}
\begin{proof}
Easy induction on the derivation of $t$.
\end{proof}

If $A$ is a theorem, then by induction on the length of Hilbert derivations $\cdot\vdash
A$ is derivable in the type system, so, by the lemma above, $\lock\vdash A$, so by the
$\shut$ rule $\cdot\vdash\Box A$.

We then turn to soundness.

\begin{lemma}\label{lem:sound_lem}
$\den{\Gamma\vdash A\to B}\to\den{\Gamma\vdash A}\to\den{\Gamma\vdash B}$
is a theorem of $IK$.
\end{lemma}
\begin{proof}
By induction on $\Gamma$. The base case is trivial. The variable case asks that
\[
  (C\to\den{\Gamma\vdash A\to B})\to(C\to\den{\Gamma\vdash A})\to C\to
  \den{\Gamma\vdash B}.
\]
This follows by assuming all formulae to the left of implications, applying Modus
Ponens twice, and induction.

The lock case starts by applying necessitation to the induction hypothesis, then using
the K axiom to distribute the box through.
\end{proof}

\begin{lemma}
All typing rules are sound with respect to the formula translation. 
\end{lemma}
\begin{proof}
The $\lambda$ and $\shut$ rules are trivial because in each the formula translations
of the premise and conclusion are identical.

Variable rule: Let $\Gamma'$ (which contains no locks) be $B_1,\ldots,B_n$. Then
$A\to B_1\to\cdots\to B_n\to A$ is a theorem. We then construct the context
$\Gamma$ from the right by, for formulae $B$, observing that $B\to C$ is a theorem
for any theorem $C$, and, for locks, using necessitation.

Application follows by Lemma~\ref{lem:sound_lem} and Modus Ponens.

$\open$: note that $\Gamma,\lock,B_1,\ldots,B_n\vdash A$ has the same
interpretation as $\Gamma\vdash\Box(B_1\to\cdots\to B_n\to A)$. Now $A\to(B_1\to
\cdots\to B_n\to A)$ is a theorem; by necessitation and K, so is $\Box A\to\Box(B_1\to
\cdots\to B_n\to A)$. Then $\den{\Gamma\vdash\Box A\to\Box(B_1\to\cdots\to B_n\to
A)}$ is likewise. Lemma~\ref{lem:sound_lem} completes the proof.
\end{proof}

As a corollary, any type inhabited in the empty context, i.e. $\cdot\vdash A$, is indeed
a theorem of $IK$.

\subsection{Proof of Theorem~\ref{thm:sn} (Strong Normalisation)}\label{SN_app}

Strong normalisation could be proved in a number of ways; we choose Tait's method,
as presented for example by Girard et al.~\cite[Chapter 6]{Girard:Proofs}. We define
sets $\RED{A}$ of \emph{reducible} untyped terms by induction on the type $A$ by
taking the usual definitions, e.g.
\begin{itemize}
\item $t\in\RED{A\to B}$ if for all $u\in\RED{A}$, $t\,u\in\RED{B}$
\end{itemize}
and extending them with
\begin{itemize}
\item $t\in\RED{\Box A}$ if $\open\,t\in\RED{A}$.
\end{itemize}

A term is \emph{neutral} if it is a variable $x$, or if its outermost term-former is an
elimination, e.g. it has form $t\,u$ or $\open\,t$.

\begin{lemma}
All sets $\RED{A}$ obey the criteria
\begin{description}
\item[(CR1)] If $t\in\RED{A}$ then $t$ is normalisable;
\item[(CR2)] If $t\in\RED{A}$ and $t\red t'$ then $t'\in\RED{A}$;
\item[(CR3)] If $t$ is neutral and for all one step reductions $t'$ of $t$ we have
$t'\in\RED{A}$, then $t\in\RED{A}$.
\end{description}
Note that the third criteria vacuously implies that if $t$ is neutral and normal, then it is
reducible. We call this criterion (CR4).
\end{lemma}
\begin{proof}
We prove all three properties simultaneously by induction on the type. Here we
present only the $\Box A$ case.

For (CR1), the reduction sequences starting with $t\in\RED{\Box A}$ are in
correspondence with the reduction sequences on $\open\,t$ which do not touch the 
outer $\open$. By induction these may only be finitely long, so $t$ is normalisable.

For (CR2), if $t\in\RED{\Box A}$ then $\open\,t\in\RED{A}$, and if $t\red t'$ then
$\open\,t\red\open\,t'$ and so by induction $\open\,t'\in\RED{A}$ as required.

For (CR3), if $t$ is neutral then all reductions $\open\,t\mapsto\open\,t'$ are in
correspondence with reductions $t\mapsto t'$ because $t$ does not have $\shut$ as
its outermost term-former. If we have $t'\in\RED{\Box A}$ for all
such reductions then $\open\,t'\in\RED{A}$ by definition. Hence by induction
$\open\,t\in\RED{A}$ as required.
\end{proof}

\begin{lemma}\label{lem:values_red}
If $t\in\RED{A}$ then $\shut\,t\in\RED{\Box A}$.
\end{lemma}
\begin{proof}
We need only that $\open\,\shut\,t\in\RED{A}$, which we show by induction on
$\nu(t)$. $\open\,\shut\,t$ reduces to $t$, which is in $\RED{A}$. All other
reductions of this term have form $\open\,\shut\,t'$, where $t\mapsto t'$. But $\nu(t)
>\nu(t')$, so by induction $\open\,\shut\,t'\in\RED{A}$, so (CR3) concludes the
proof.
\end{proof}

\begin{lemma}\label{lem:fundamental}
Let $\Gamma\vdash t:A$ be a typed term where $\Gamma$ has as variable
assignments $x_1:A_1,\ldots,x_n:A_n$, and let $u_1,\ldots,u_n$ be a set of terms
with $u_i\in\RED{A_i}$ for $1\leq i\leq n$. Then $t[u_1/x_1,\ldots,u_n/x_n]\in
\RED{A}$.
\end{lemma}
\begin{proof}
By induction on the derivation of $t$. Looking only at the term-formers for $\Box$, this
holds for $\shut\,t$ by induction and Lemma~\ref{lem:values_red}, and for $\open\,t$
by induction and the definition of $\RED{\Box A}$.
\end{proof}

Theorem~\ref{thm:sn} may hence be proved as follows: variables are neutral and
normal, so are in $\RED{A_i}$ for any $A_i$ by (CR4). Hence we can apply
Lemma~\ref{lem:fundamental} to the identity substitution, replacing variables by
themselves, to conclude that $t\in\RED{A}$. Then by (CR1) $t$ is normalisable.

\subsection{Proof of Theorem~\ref{thm:can} (Canonicity)}

\begin{lemma}\label{lem:normal_neutral_open}
If an typed term-in-context $\Gamma\vdash t:A$ is normal and neutral, then $t$
contains a free variable.
\end{lemma}
\begin{proof}
By induction on the derivation of $t$. We present only the $\open\,t$ case: if such a
term is normal then $t$ is normal and does not start with $\shut$, and so is neutral,
and so by induction contains a variable.
\end{proof}

Theorem~\ref{thm:can} then follows because a normal term with no free variables
cannot be neutral by the lemma above, so it must have as main term-former the
appropriate introduction.

\subsection{Proof of Theorem~\ref{thm:subform} (Subformula Property)}

\begin{lemma}\label{lem:sf_lem}
If $\Gamma\vdash t:A$ is normal and neutral then $A$ is a subtype of some type
assigned in $\Gamma$.
\end{lemma}
\begin{proof}
By induction on $t$, as usual presenting only the $\Box$ case:

If $\open\,t$ is normal and neutral then so must $t$ be. But then by induction the type
$\Box A$ of $t$ is a subformula of $\Gamma$, so the type $A$ of $\open\,t$ is also.
\end{proof}

We then prove Theorem~\ref{thm:subform} by induction on $t$, presenting only the
$\Box$ cases:

$\Gamma\vdash\shut\,t:\Box A$: all proper subterms are subterms of $t$, which by
induction have type included in $\Gamma$, or $A$, and hence, in the latter case,
$\Box A$.

$\Gamma,\lock,\Gamma'\vdash\open\,t:A$: by induction all subterms of $t$ are
contained in $\Gamma$ or $\Box A$. But if $\open\,t$ is normal then $t$ is neutral, so
by Lemma~\ref{lem:sf_lem} $\Box A$ is contained in $\Gamma$.

\subsection{Proof of Theorem~\ref{thm:IKcatsnd} (Categorical Soundness)}

We first confirm the soundness of $\open\,\shut\,t\mapsto t$.
$\den{\Gamma,\lock,\Gamma'\vdash\open\,\shut\,t:A}$ is
$\varphi^{-1}\varphi\,\den{\Gamma,\lock\vdash t:A}$ composed with a projection,
where $\varphi$ is the isomorphism given by the $\Dia\dashv\Box$ adjunction. But
this is exactly $\den{\Gamma,\lock,\Gamma'\vdash t:A}$ as required.

The soundness of $\beta$-reduction for functions follows immediately given the
lemma below regarding the interpretation of substitution. The lemma is slightly more
general than necessary for function spaces, because the general form will be useful
later.

\begin{lemma}\label{lem:catsubs}
If $\Gamma,x:A,\Gamma'\vdash t:B$ and $\Gamma\vdash u:A$ then
$\den{\Gamma,\Gamma'\vdash t[u/x]:B}$ is
\[\xymatrix{
  \Gamma,\Gamma' \ar[r]^-{\Gamma'\langle\Gamma,u\rangle} & \Gamma,A,\Gamma' \ar[r]^-{t} & B
}\]
\end{lemma}
\begin{proof}
By induction on $t$. We present one variable case and the cases for $\Box$; other
cases are routine.

If $t$ is the variable $x$ then $\Gamma'$ contains no locks so we have
\[\xymatrix{
 \Gamma,\Gamma' \ar[r]^-{\Gamma'\langle\Gamma,u\rangle} \ar[d]_{pr} & \Gamma,A,\Gamma' \ar[d]_{pr} \ar[ddr]^-{pr} \\
 \Gamma \ar[r]^-{\langle\Gamma,u\rangle} \ar[drr]_-u & \Gamma,A \ar[dr]_(0.35){pr} \\
 & & A
}\]

$\shut$:
\[\xymatrix{
  \Gamma,\Gamma'  \ar[r]^-{\Gamma'\langle\Gamma,u\rangle} \ar[d]_{\eta} & \Gamma,A,\Gamma' \ar[d]^{\eta} \\
  \Box\Dia(\Gamma,\Gamma') \ar[r]_-{\Box\Dia\Gamma'\langle\Gamma,u\rangle} & \Box\Dia(\Gamma,A,\Gamma') \ar[r]_-{\Box t} & \Box B
}\]

$\open$:
\[\xymatrix{
  \Gamma,\Gamma',\lock,\Gamma'' \ar[rr]^-{(\Gamma,\sslock,\Gamma'')\langle\Gamma,u\rangle} \ar[d]_{pr} & & \Gamma,A,\Gamma',\lock,\Gamma'' \ar[d]^{pr} \\
  \Gamma,\Gamma,\lock \ar[rr]_-{\Dia\Gamma'\langle\Gamma,u\rangle} & & \Gamma,A,\Gamma',\lock \ar[r]_-{\Dia t} & \Dia\Box B \ar[r]_-{\varepsilon} & B
}\]
\end{proof}

Finally, the soundness of $\eta$-equivalence for $\Box$ follows immediately from the
$\Dia\dashv\Box$ adjunction.

\subsection{Sums in $IK$}\label{app:sums}

We finally show that the proofs above still hold given the generalised rules for sums
of Figure~\ref{fig:sums}.

For the logical soundness of the $\casebare$ rule we note that
\[
A+B\to(A\to\den{\Gamma'\vdash C})\to(B\to\den{\Gamma'\vdash C})\to\den{\Gamma'\vdash C}
\]
is a theorem and use Lemma~\ref{lem:sound_lem} and the $K$ axiom to complete.
For $\abort$ we use that $0\to\den{\Gamma'\vdash A}$.

The other syntactic proofs proceed as usual for sums.

For the categorical semantics, we need to interpret the new $\casebare$ and $\abort$
rules, and to confirm that the $\beta$-reductions and commuting conversions still
hold.

$\casebare$: We interpret $\case{s}{x}{t}{y}{u}$ as
\[\xymatrix{
  \Gamma,\Gamma' \ar[r]^-{\Gamma'\langle\Gamma,s\rangle} &
  \Gamma,A+B,\Gamma' \ar[r]^-d &
  (\Gamma,A,\Gamma')+(\Gamma,B,\Gamma') \ar[r]^-{[t,u]} & C
}\]
Products and $\Dia$ are left adjoints, so preserve coproducts, yielding the
isomorphism $d$.

Now where $\Gamma'$ is empty we can show by diagram chase that
$d\circ(\Gamma\times in_1)=in_1:\Gamma\times A\to
(\Gamma\times A)+(\Gamma\times B)$. This suffices to show that
$\beta$-reduction (for the first injection) is sound where $\Gamma'$ is empty. For
the general $\Gamma'$ we hence use this as the base case for a proof that
$d\circ\Gamma'(\Gamma\times in_1)=in_1:
\Gamma,A,\Gamma'\to\Gamma,A,\Gamma'+
\Gamma,B,\Gamma'$. We show the step case for locks only:
\[\xymatrix{
  \Gamma,A,\Gamma',\lock \ar[r]^-{id} \ar[d]_{\Dia\Gamma'(id\times in_1)} \ar[dr]^-{\Dia\eta} \ar[ddr]^>>>>{\Dia in_1} & \Gamma,A,\Gamma',\lock \ar[r]^-{in_1} &\Gamma,A,\Gamma',\lock\,+\,\Gamma,B,\Gamma',\lock \\
  \Gamma,A\times B,\Gamma',\lock \ar[dr]_-{\Dia d} & \Dia\Box(\Gamma,A,\Gamma',\lock) \ar[u]_{\varepsilon} \ar[dr]^-{\Dia\Box in_1} \\
  & \Dia(\Gamma,A,\Gamma'\,+\,\Gamma,B,\Gamma') \ar[r]_-*!/u-4pt/{\labelstyle\Dia[\Box in_1\circ\eta,\Box in_2\circ\eta]} & \Dia\Box(\Gamma,A,\Gamma',\lock\,+\,\Gamma,B,\Gamma',\lock) \ar[uu]_{\varepsilon}
}\]
Here the leftmost triangle commutes by induction.

The commuting conversions can be confirmed by diagram chase similarly. We
present one case:
\[\xymatrix{
  \Gamma,\Gamma',\lock,\Gamma'' \ar[r]^-*!/u4pt/{\labelstyle(\Gamma',\sslock,\Gamma'')\langle\Gamma,s\rangle} \ar[d]_{pr} & \Gamma,A+B,\Gamma',\lock,\Gamma'' \ar[r]^-d \ar[d]^{pr} & \Gamma,A,\Gamma',\lock,\Gamma''\,+\,\Gamma,B,\Gamma',\lock,\Gamma'' \ar[r]^-*!/u4pt/{\labelstyle[\varepsilon\circ\Dia t\circ pr,\varepsilon\circ\Dia u\circ pr]} \ar[d]_{[pr,pr]} \ar[dr]_(0.42){[\Dia t\circ pr,\Dia u\circ pr]\;\;} & C \\
  \Gamma,\Gamma',\lock \ar[r]_-*!/u-4pt/{\labelstyle\Dia\Gamma'\langle\Gamma,s\rangle} & \Gamma,A+B,\Gamma',\lock \ar[dr]_-{\Dia d} & \Gamma,A,\Gamma',\lock\,+\,\Gamma,B,\Gamma',\lock \ar[r]_(0.66){\Dia t+\Dia u} \ar[d]_{\cong} & \Dia\Box C \ar[u]_{\varepsilon} \\
  & & \Dia(\Gamma,A,\Gamma'\,+\,\Gamma,B,\Gamma') \ar[ur]_-{\Dia[t,u]}
}\]
Here the top line is $\den{\case{s}{x}{\open\,t}{y}{\open\,u}}$, while the other
perimeter is $\den{\open\,\case{s}{x}{t}{y}{u}}$.

$\abort$ is interpreted as
\[\xymatrix{
  \Gamma,\Gamma' \ar[r]^-{\Gamma' t} & 0,\Gamma' \ar[r]^-{!} & A
}\]
where the right hand arrow is unique because products and $\Dia$ are left adjoints,
so preserve the initial object. The soundness of the commuting conversions then
follow easily. We show one case:
\[\xymatrix{
  \Gamma,\Gamma',\lock,\Gamma'' \ar[rr]^-{(\Gamma',\sslock,\Gamma')t} \ar[d]_{pr} & & 0,\Gamma',\lock,\Gamma'' \ar[r]^-{!} \ar[d]_{pr} & A \\
  \Gamma,\Gamma',\lock \ar[rr]_-{\Dia\Gamma' t} & & 0,\Gamma',\lock \ar[r]_-{\Dia !} & \Dia\Box A \ar[u]_{\varepsilon}
}\]
The top line is $\den{\abort\,t}$ and the other perimeter is $\den{\open\,\abort\,t}$.


\section{Left Adjoints and Categorical Completeness}

\subsection{Type System and Computation}

\textbf{Logical soundness}: For $\dia$, we have $\Dia A\to B_1\to\cdots\to B_n\to
\Dia A$, which by necessitation and K yields $\Box\Dia A\to\Box(B_1\to\cdots\to B_n
\to\Dia A)$. But $A\to\Box\Dia A$ by $\eta$. Using Lemma~\ref{lem:sound_lem} for
the context $\Gamma$ completes the proof.

For $\letdbare$, the second premise yields $\Dia A\to B$ by monotonicity and
$\varepsilon$. Then $\den{\Gamma\vdash\Dia A\to B}$, and
Lemma~\ref{lem:sound_lem} completes the proof.

\textbf{Subject reduction}: We need left weakening (Lemma~\ref{lem:left_weak})
and variable weakening to weaken $x:A,\lock\vdash u:B$ to $\Gamma,x:A,\lock,
\Gamma'\vdash u:B$, and may then apply the substitution.

\textbf{Strong normalisation}: The $\dia$ rule involves a `parasitic' type, as with
sums, so we use similar techniques as for sums to extend Tait's method.
Unfortunately these techniques appear to be folklore and we could not find an
explicit description in the literature%
\footnote{But see \url{mathoverflow.net/questions/281387}.}%
, so we write out the proof with some care. We set
\begin{itemize}
\item $t\in\RED{\Dia A}$ if $t\red\dia\,u$ for $u\in\RED{A}$, or $t$ normalises to a
neutral term.
\end{itemize}

For (CR1), if $t\red\dia\,u$ and $u$ is by induction normalisable, then $\dia\,u$ is
also normalisable because no reduction touches the outer $\dia$. Otherwise $t$ is
normalisable by definition.

For (CR2), if $t\red t'$ and $t\red\dia\,u$ then by confluence $t'\red\dia\,u'$ for some
$u\red u'$. By induction $u'\in\RED{A}$, so $t'\in\RED{\Dia A}$. Else if $t$ normalises
to a neutral term then so does $t'$.

For (CR3), if there exists a one-step reduction of $t$ that reduces in turn to some
$\dia\,u$ then $t$ does also. Else if all one-step reductions of $t$ reduce to a neutral
term then $t$ does similarly.

Lemma~\ref{lem:fundamental} can then be extended to the new term-formers. The
$\dia$ case follows immediately by definition. For $\letdbare$ we use a secondary
induction on $\nu(t)+\nu(u)$, and (CR3). We omit the substitutions for $\Gamma$ in
the below for clarity. If $t$ has form $\dia\,s$ then one possible
reduction is to $u[s/x]$. But by definition $s\in\RED{A}$, so $u[s/x]\in\RED{B}$. If
$\nu(t)+\nu(u)=0$ this is the only possible reduction. Otherwise
we might reduce one of the subterms $t$ or $u$; without loss of generality, say
$t\mapsto t'$. By (CR2) $t'\in\RED{\Dia A}$. But then we can use our secondary
induction to conclude $\letd{x}{t'}{u}\in\RED{B}$.

\textbf{Canonicity} follows as before.

\subsection{Categorical Semantics}

To determine $\beta$-reduction for function spaces, which involves substitution, still
holds we must confirm that Lemma~\ref{lem:catsubs} extends to the new
type-formers for $\Dia$, which is straightforward from expanding the definitions.

$\beta$-reduction for $\Dia$ involves a left weakening \`{a} la
Lemma~\ref{lem:left_weak}, so we must determine the categorical equivalent of this.

\begin{lemma}\label{lem:cat_leftweak}
Given $\Gamma'\vdash t:A$, the categorical denotation of $\Gamma,\Gamma'\vdash
t:A$ is $\Gamma'\vdash t:A\circ\Gamma'!$, where $!$ is the unique arrow
$\Gamma\to 1$ (noting that $\Gamma'\cong\Gamma'1$; we abuse notation by 
treating them as equal).
\end{lemma}
\begin{proof}
By induction on $t$. The variable case holds because we use the projection to $A$.

For $\lambda$ we use the diagram below. The natural $\eta^c$ is the unit of the
cartesian closure adjunction, while the triangle commutes by induction.
\[\xymatrix{
  \Gamma,\Gamma' \ar[d]_{\Gamma'!} \ar[r]^-{\eta^c} & A\to(\Gamma,\Gamma',A) \ar[r]^-{A\to t} \ar[d]_{A\to(\Gamma'!\times A)} & A\to B \\
  \Gamma' \ar[r]_-{\eta^c} & A\to(\Gamma',A) \ar[ur]_-{A\to t}
}\]

Application: The denotation of the weakened term is $\varepsilon^c\circ
\langle t,u\rangle$, where $\varepsilon^c$ is the counit of the cartesian closure
adjunction. This is $\varepsilon^c\circ\langle t,u\rangle\circ\Gamma'!$ by induction.

$\shut$:
\[\xymatrix{
  \Gamma,\Gamma' \ar[d]_{\Gamma'!} \ar[r]^-{\eta^m} & \Box(\Gamma,\Gamma',\lock) \ar[r]^-{\Box t} \ar[d]_{\Box\Dia\Gamma'!} & \Box A \\
  \Gamma' \ar[r]_-{\eta^m} & \Box(\Gamma',\lock) \ar[ur]_-{\Box t}
}\]

$\open$:
\[\xymatrix{
  \Gamma,\Gamma',\lock,\Gamma'' \ar[d]_{(\Gamma',\sslock,\Gamma'')!} \ar[r]^-{pr} & \Gamma,\Gamma',\lock \ar[r]^-{\Dia t} \ar[d]_{\Dia\Gamma'!} & \Dia\Box A \ar[r]^-{\varepsilon^{m}} & A \\
  \Gamma',\lock,\Gamma'' \ar[r]_-{pr} & \Gamma',\lock \ar[ur]_-{\Dia t}
}\]

$\dia$:
\[\xymatrix{
  \Gamma,\Gamma',\lock,\Gamma'' \ar[d]_{(\Gamma',\sslock,\Gamma'')!} \ar[r]^-{pr} & \Gamma,\Gamma',\lock \ar[r]^-{\Dia t} \ar[d]_{\Dia\Gamma'!} & \Dia A \\
  \Gamma',\lock,\Gamma'' \ar[r]_-{pr} & \Gamma',\lock \ar[ur]_-{\Dia t}
}\]

$\letdbare$ follows immediately from definition.
\end{proof}

This lemma, along with Lemma~\ref{lem:catsubs}, establishes for us that $u[t/x]$ is
$u\circ\Dia(!\times A)\circ\Dia\langle\Gamma,t\rangle\circ pr$. The middle two
arrows simplify to $\Dia t$, as required by $\letd{x}{\dia\,t}{u}$.

For $\eta$-equivalence $\letd{x}{t}{x}$ has denotation $t$ by definition. For
the associativity equivalence it clear that both sides equal $(t\circ u)\circ s=
t\circ(u\circ s)$. \\

We now move to the \textbf{term model construction}.

$\Box$ is a functor: $\Box$ applied to the identity $x$ is $\shut\,\open\,x$, which is
$\eta$-equal to $x$. The composition of $\Box u$ and $\Box t$ is $\shut\,u[\open\,
\shut\,t[\open\,x/x]/x]$, which reduces to $\shut\,u[t[\open\,x/x]/x]$ as required.

$\Dia$ is a functor: $\Dia x$ is $\letd{x}{x}{\dia\,x}$, which is $\eta$-equal to the
identity. $\Dia u\circ\Dia t$ is $\letd{x}{(\letd{x}{x}{\dia\,t})}{\dia\,u}$. By associativity
this equals $\letd{x}{x}{\letd{x}{\dia\,t}{\dia\,u}}$, which reduces to
$\letd{x}{x}{(\dia\,u[t/x])}$ as required.

$\eta^m$ is natural: We require that $\Box\Dia t\circ\eta^m=\eta^m\circ t$. The
left hand side is $\shut\,\letd{x}{\open\,\shut\,\dia\,x}{\dia\,t}$, which reduces to
$\shut\,\letd{x}{\dia\,x}{\dia\,t}$, then to $\shut\,\dia\,t$.

$\varepsilon^m$ is natural: We require that $\varepsilon^m\circ\Dia\Box t=t\circ
\varepsilon^m$. The left hand side is
\[
\letd{y}{(\letd{x}{x}{\dia\,\shut\,t[\open\,x/x]})}{\open\,y}
\]
By associativity this equals
\[
\letd{x}{x}{\letd{y}{\dia\,\shut\,t[\open\,x/x]}{\open\,y}}
\]
which reduces to $\letd{x}{x}{(\open\,\shut\,t[\open\,x/x])}$, then in turn reduces to
$\letd{x}{x}{(t[\open\,x/x])}$. But by associativity this is $t[\letd{x}{x}{\open\,x}/x]$
as required.

Triangle equalities: We first check that $\Box\varepsilon^m\circ\eta^m$ is the identity
on any $\Box A$. This is $\shut\,\letd{x}{\open\,\shut\,\dia\,x}{\open\,x}$, which
reduces in turn to $\shut\,\letd{x}{\dia\,x}{\open\,x}\mapsto\shut\,\open\,x$, which is
$\eta$-equivalent to $x$.

We then check that $\varepsilon^m\circ\Dia\eta^m$ is the identity on any $\Dia A$.
This is
\[
\letd{x}{(\letd{x}{x}{\dia\,\shut\,\dia\,x})}{\open\,x}\
\]
which by associativity equals $\letd{x}{x}{\letd{x}{\dia\,\shut\,\dia\,x}{\open\,x}}$,
which reduces to $\letd{x}{x}{\open\,\shut\,\dia\,x}\mapsto\letd{x}{x}{\dia\,x}$,
which is $\eta$-equal to $x$.

\textbf{Proof of Lemma~\ref{lem:compl_lem}.}

$\Gamma,x:A,y_1:B_1,\cdots,y_n:B_n\vdash x:A$: The term model denotation
$\den{x}$ is $x:\den{\Gamma}\vdash\pi_2\pi_1\cdots\pi_1\,x:A$, where there are
$n$ first projections. The context term is $\Gamma\vdash\langle\langle\langle
c_{\Gamma},x\rangle,y_1\rangle,\ldots,y_n\rangle:\den{\Gamma}$. The substitution
applied to the term is then $\pi_2\pi_1\cdots\pi_1\langle\langle\langle c_{\Gamma},x\rangle,y_1\rangle,
\ldots,y_n\rangle$, which reduces to $x$.

$\den{\lambda y.t}$ is $\lambda y.\den{t}[\langle x,y\rangle/x]$. But then the
substitution is $\lambda y.\den{t}[\langle c_{\Gamma},y\rangle/x]$, which equals
$\lambda y.t$ by induction.

$\den{tu}$ is $\den{t}\den{u}$, so the substitution is $(\den{t}[c_{\Gamma}/x])
(\den{u}[c_{\Gamma}/x])$, which equals $tu$ by induction.

$\den{\shut\,t}$ is $\Box\den{t}\circ\eta^m$, which is $\shut\den{t}[\open\,\shut\,
\dia\,x/x]$, which reduces to $\shut\den{t}[\dia\,x/x]$. The substitution is then
$\shut\den{t}[\dia\,c_{\Gamma}/x]$, which equals $\shut\,t$ by induction.

$\den{\open\,t}$ is $\varepsilon^m\circ\Dia t\circ pr$, where $pr$ is the projection
out of the weakening of the variables $y_1,\ldots,y_n$. This is
\[
  \letd{y}{(\letd{x}{\pi_1\cdots\pi_1 x}{\dia\den{t}})}{\open\,y}
\]
where there are $n$ first projections. By associativity this equals
\[
  \letd{x}{\pi_1\cdots\pi_1 x}{\letd{y}{\dia\den{t}}{\open\,y}}
\]
which reduces to $\letd{x}{\pi_1\cdots\pi_1 x}{\open\den{t}}$. The substitution is
then
\[
  \letd{x}{\pi_1\cdots\pi_1\langle\langle\dia\,c_{\Gamma},y_1\rangle,\ldots,y_n\rangle}{\open\den{t}}
\]
This reduces to $\letd{x}{\dia\,c_{\Gamma}}{\open\den{t}}\mapsto\open\den{t}
[c_{\Gamma}/x]$, which equals $\open\,t$ by induction.

$\den{\dia\,t}$ is $\Dia\den{t}\circ pr$, which is
$\letd{x}{\pi_1\cdots\pi_1\,x}{\dia\den{t}}$. The substitution is
\[
  \letd{x}{\pi_1\cdots\pi_1\,\langle\langle\dia\,c_{\Gamma},y_1\rangle,\ldots,y_n\rangle}{\dia\den{t}}
\]
which reduces to $\letd{x}{\dia\,c_{\Gamma}}{\dia\den{t}}\mapsto
\dia\den{t}[c_{\Gamma}/x]$, which equals $\dia\,t$ by induction.

$\den{\letd{x}{t}{u}}$ is $\den{u}\circ\den{t}$, so the substitution is $\den{u}
[\den{t}[c_{\Gamma}/x]/x]$, which by induction equals $\den{u}[t/x]$. This is
$\eta$-equal to $\den{u}[\letd{x}{t}{\dia\,x}/x]$, which is by associativity equal
to $\letd{x}{t}{(\den{u}[\dia\,x/x])}$. This equals $\letd{x}{t}{u}$ by induction.


\section{Intuitionistic S4}

\subsection{Type System and Computation}

Logical soundness follows by showing that $\Box A\to
\den{\Gamma'\vdash A}$ is a theorem by induction on $\Gamma'$, then using
Lemma~\ref{lem:sound_lem} to incorporate $\Gamma$, and finally noting that
$\den{\Gamma\vdash\den{\Gamma'\vdash A}}=\den{\Gamma,\Gamma'\vdash A}$.

For the induction, the base case is the T axiom. The variable case extends the
induction hypothesis $\Box A\to\den{\Gamma'\vdash A}$ to $\Box A\to B\to
\den{\Gamma'\vdash A}$. The lock case combines the 4 axiom with $\Box\Box A\to
\Box\den{\Gamma'\vdash A}$, which follows by applying the K axiom to the induction.

\subsection{Categorical Semantics}

We here show that Lemma~\ref{lem:catsubs} extends to the new $\open$ rule. If
the variable substituted for is part of the weakening, this is easy. Suppose instead that
we have $\open$ applied to $\Gamma,x:A,\Gamma'\vdash t:\Box B$ with weakening
$\Gamma''$. Then
\[\xymatrix{
  \Gamma,\Gamma',\Gamma'' \ar[r]^-{l_{\Gamma''}} \ar[d]_{(\Gamma',\Gamma'')\langle\Gamma,u\rangle} & \Gamma,\Gamma',\lock \ar[d]^{(\Gamma',\sslock)\langle\Gamma,u\rangle} \\
  \Gamma,A,\Gamma',\Gamma'' \ar[r]_-{l_{\Gamma''}} & \Gamma,A,\Gamma',\lock \ar[r]_-{\Dia t} & \Dia\Box B \ar[r]_-{\varepsilon^m} & B
}\]
where the square commutes by the naturality of $l_{\Gamma''}$.

\textbf{Proof of Lemma~\ref{lem:lock_over}:}

The base case of the induction on $\Gamma_3$ follows because
$l_{\Gamma_2,\sslock}\circ l_{\cdot}$ is $\mu\circ\Dia l_{\Gamma_2}\circ\eta$,
which is $\mu\circ\eta\circ  l_{\Gamma_2}$ by the naturality of $\eta$,
which is $l_{\Gamma_2}$ by the monad laws.

The variable case follows because $l_{\Gamma_2,\sslock}\circ l_{\Gamma_3,A}$ is
$l_{\Gamma_2,\sslock}\circ l_{\Gamma_3}\circ pr=l_{\Gamma_2,\Gamma_3}\circ
pr$ by induction, which is $l_{\Gamma_2,\Gamma_3,A}$.

The lock case has $l_{\Gamma_2,\sslock}\circ l_{\Gamma_3,\sslock}=
\mu\circ\Dia l_{\Gamma_2}\circ\mu\circ\Dia l_{\Gamma_3}=
\mu\circ\mu\circ\Dia\Dia  l_{\Gamma_2}\circ\Dia l_{\Gamma_3}$ by the naturality of
$\mu$, which is $\mu\circ\Dia\mu\circ\Dia\Dia  l_{\Gamma_2}\circ
\Dia l_{\Gamma_3}$ by the monad laws, which is
$\mu\circ\Dia l_{\Gamma_2,\sslock}\circ \Dia l_{\Gamma_3}=
\mu\circ l_{\Gamma_2,\Gamma_3}$ by induction, which is
$l_{\Gamma_1,\Gamma_2,\sslock}$.

The step cases of the main induction follow by easy diagram chases.

\textbf{Proof of Lemma~\ref{lem:lockreponterms}:}

By induction on the derivation of $t$. The variable case holds because we use the
projection to $A$.

$\lambda$: By induction
$\Gamma,\lock,\Gamma'',A\vdash t:B\circ(\Gamma''(l_{\Gamma'})\times A)=
\Gamma,\Gamma',\Gamma'',x:A\vdash t:B$. This yields the triangle of
\[\xymatrix{
  \Gamma,\Gamma',\Gamma'' \ar[r]^-{\eta^c} \ar[d]_{\Gamma''(l_{\Gamma'})} & A\to\Gamma,\Gamma',\Gamma'',A \ar[r]^-{A\to t} \ar[d]_{A\to(\Gamma''(l_{\Gamma'})\times A)} & A\to B \\
  \Gamma,\lock,\Gamma'' \ar[r]_-{\eta^c} & A\to\Gamma,\lock,\Gamma',A \ar[ur]_-{A\to t}
}\]
where the natural $\eta^c$ is the unit of the cartesian closure adjunction.

Application: $\varepsilon^c\circ\langle t,u\rangle\circ\Gamma''(l_{\Gamma'})$,
where $\varepsilon^c$ is the counit of the cartesian closure adjunction, equals
$\varepsilon^c\circ\langle t\circ\Gamma''(l_{\Gamma'}),u\circ
\Gamma''(l_{\Gamma'})\rangle$, which is $\varepsilon^c\circ\langle t,u\rangle$ by
induction as required.

$\shut$: By induction $\Gamma,\lock,\Gamma'',\lock\vdash t:A\circ
\Dia\Gamma''(l_{\Gamma'})=\Gamma,\Gamma',\Gamma'',\lock\vdash t:A$. This
yields the triangle of
\[\xymatrix{
  \Gamma,\Gamma',\Gamma'' \ar[r]^-{\eta^{m}} \ar[d]_{\Gamma''(l_{\Gamma'})} & \Box(\Gamma,\Gamma',\Gamma'',\lock) \ar[r]^-{\Box t} \ar[d]_{\Box\Dia \Gamma''(l_{\Gamma'})} & \Box A \\
  \Gamma,\lock,\Gamma'' \ar[r]_-{\eta^{m}} & \Box(\Gamma,\lock,\Gamma',\lock) \ar[ur]_-{\Box t}
}\]
where the natural $\eta^{m}$ is the unit of the modal adjunction $\Dia\dashv\Box$.

$\open$: Where the lock in question is part of the context introduced by the
weakening we use Lemma~\ref{lem:lock_over}. Where it was part of the original
context we use induction and the naturality of lock replacement.

\subsection{Coherence}

We first note that for the proof of the first part of Lemma~\ref{lem:strength}, in
the case where $\open$ introduces $x$ as part of the weakening, we need that
$l_{\Gamma',x:A,\Gamma''}=l_{\Gamma',\Gamma''}\circ\Gamma''(pr)$. This is
easily proved by induction on $\Gamma''$. Similarly, for the second part of the lemma
in the case that $\open$ introduces the lock, we need $l_{\Gamma',\Gamma''}=
l_{\Gamma',\sslock,\Gamma''}\circ\Gamma''(\eta)$. This follows by induction on
$\Gamma''$, with the base case using the naturality of $\eta$ and the monad laws.

We present the variable case of Lemma~\ref{lem:strength_and_l}:

\[\xymatrix{
  \Gamma,\Gamma',B \ar[rr]^-{t} \ar[dd]_{l_{\Gamma',B}} \ar[dr]_-{pr} & & A \ar[dd]^{\eta} \\
  & \Gamma,\Gamma' \ar[ur]_-{t} \ar[dl]_-{l_{\Gamma'}} \\
  \Gamma,\lock \ar[rr]_-{\Dia t} & & \Dia A
}\]
The top triangle commutes by Lemma~\ref{lem:strength}, the left by
definition, and the bottom-right by induction.

Finally we present the step cases of Lemma~\ref{lem:strength_and_l2}:

For the variable case, the bottom-right triangle commutes by induction in:
\[\xymatrix{
  \Gamma,\Gamma',\Gamma'',B \ar[rr]^-{l_{\Gamma'',B}} \ar[dd]_{l_{\Gamma',\Gamma'',B}} \ar[dr]^-{pr}  & & \Gamma,\Gamma',\lock \ar[dd]^{\Dia t} \\
  & \Gamma,\Gamma',\Gamma'' \ar[ur]_-{l_{\Gamma''}} \ar[dl]^-{\,l_{\Gamma',\Gamma''}} \\
  \Gamma,\lock \ar[rr]_-{\Dia t} & & \Dia A
}\]

Lock case:
\[\xymatrix{
  \Gamma,\Gamma',\Gamma'',\lock \ar[rrr]^-{l_{\Gamma'',\tlock}} \ar[ddd]_{l_{\Gamma',\Gamma'',\tlock}} \ar[drr]^-{\Dia l_{\Gamma''}} \ar[ddr]^(0.6){\Dia l_{\Gamma',\Gamma''}} & & & \Gamma,\Gamma',\lock \ar[ddd]^{\Dia t} \\
  & & \Gamma,\Gamma',\lock,\lock \ar[ur]^-{\mu} \ar[d]^{\Dia\Dia t} \\
  & \Gamma,\lock,\lock \ar[dl]^-{\mu} \ar[r]_{\Dia\Dia t} & \Dia\Dia A \ar[dr]_-{\mu} \\
  \Dia\Gamma \ar[rrr]_-{\Dia t} & & & \Dia A
}\]
The triangles commute by definition, the trapeziums commute by the naturality of
$\mu$, and the inside quadrilateral commutes by induction.

\subsection{Left Adjoints and Categorical Completeness}

\textbf{Subject reduction:} Given $x:A,\lock\vdash u:B$ we can use left weakening to
get $\Gamma,x:A,\lock\vdash u:B$. Then given $\Gamma\vdash t:B$ we have
$\Gamma,\lock\vdash u[t/x]:B$. Then by lock replacement $\Gamma,\Gamma'\vdash
u[t/x]:B$. Hence $\letd{x}{\dia\,t}{u}\mapsto u[t/x]$ preserves the typing.

\textbf{Categorical Soundness:} We first confirm that Lemmas~\ref{lem:catsubs}
and~\ref{lem:cat_leftweak} hold for the new $\dia$ rule:
\[\xymatrix{
  \Gamma,\Gamma',\Gamma'' \ar[r]^-{l_{\Gamma''}} \ar[d]_{(\Gamma',\Gamma'')\langle\Gamma,u\rangle} & \Gamma,\Gamma',\lock \ar[d]^{(\Gamma',\sslock)\langle\Gamma,u\rangle} \\
  \Gamma,A,\Gamma',\Gamma'' \ar[r]_-{l_{\Gamma''}} & \Gamma,A,\Gamma',\lock \ar[r] \ar[r]_-{\Dia t} & \Dia B
}\]
and
\[\xymatrix{
  \Gamma,\Gamma',\Gamma'' \ar[r]^-{l_{\Gamma''}} \ar[d]_{(\Gamma',\Gamma'')!} & \Gamma,\Gamma',\lock \ar[d]^{(\Gamma',\sslock)!} \\
  \Gamma',\Gamma'' \ar[r]_-{l_{\Gamma''}} & \Gamma',\lock \ar[r] \ar[r]_-{\Dia t} & \Dia A
}\]
as $l_{\Gamma''}$ is natural. 

Then, unfolding definitions and using these two lemmas, both sides of
$\beta$-reduction for $\Dia$ have denotation
\[\xymatrix{
  \Gamma,\Gamma' \ar[r]^-{l_{\Gamma'}} & \Gamma,\lock \ar[r]^-{\Dia t} & \Dia A \ar[r]^-u & B
}\]

\textbf{Coherence:} We first must check that Lemma~\ref{lem:lockreponterms}
extends to the term-formers for $\Dia$. For $\dia$, the case where the lock is
introduced by the weakening follows by Lemma~\ref{lem:lock_over}. The other
$\dia$ case follows by induction and the naturality of lock replacement. $\letdbare$ is
easy. It is also easy to extend Lemma~\ref{lem:strength} to the new term-formers.

Then for coherence it is easy to see how to get a version of
Lemma~\ref{lem:coherence} for $\dia$.

\textbf{Categorical Completeness:} For our term model construction, we need to
confirm that $\Box$ is indeed an idempotent comonad.

We first need that $\varepsilon:\Box A\to A$, which is the term for the T axiom,
$\open\,x$, is natural, i.e. $t\circ\varepsilon=\varepsilon\circ\Box t$. The right hand
side is $\open\,\shut\,t[\open\,x/x]$, which reduces to $t[\open\,x/x]$, which is the
left hand side.

Next, $\delta:\Box A\to\Box\Box A$, which is the term for the 4 axiom,
$\shut\,\shut\,\open\,x$, is natural, i.e. $\delta\circ\Box t=\Box\Box t\circ\delta$. The
left hand side is $\shut\,\shut\,\open\,\shut\,t[\open\,x/x]$, which reduces to
$\shut\,\shut\,t[\open\,x/x]$. The right hand side is
\[
\shut\,\shut\,t[\open\,\open\,\shut\,\shut\,\open\,x/x]
\]
which reduces similarly.

Moving to the monad laws, we need that $\delta\circ\delta=\Box\delta\circ\delta$.
The left hand side is $\shut\,\shut\,\open\,\shut\,\shut\,\open\,x\mapsto
\shut\,\shut\,\shut\,\open\,x$, and right is
\[
\shut\,\shut\,\shut\,\open\,\open\,\shut\,\shut\,\open\,x
\]
reducing similarly.

We have already argued in Section~\ref{Sec:S4type} that $\varepsilon\circ\delta$ is
the identity. $\Box\varepsilon\circ\delta$ is
$\shut\,\open\,\open\,\shut\,\shut\,\open\,x$, which reduces to $\shut\,\open\,x$,
which is $\eta$-equal to $x$.

Finally, for idempotence, see the argument in Section~\ref{Sec:S4type}.

We finally prove Lemma~\ref{lem:lockrep_and_ctxterm}. First we unfold definitions
to understand the lock replacement arrows as terms:
\begin{align*}
  l_{\cdot} &= \dia\,x \\
  l_{\Gamma,A} &= l_{\Gamma}[\pi_1\,x/x] \\
  l_{\Gamma,\sslock} &= \letd{x}{x}{\open\,\letd{x}{l_{\Gamma}}{\shut\,\dia\,x}}
\end{align*}
We then proceed by induction on $\Gamma'$. The base case of the lemma has the
left hand side exactly $\dia\,c_{\Gamma}$.

$l_{\Gamma',A}[c_{\Gamma,\Gamma',y:A}/x]=l_{\Gamma'}[\pi_1
\langle c_{\Gamma,\Gamma'},y\rangle/x]$, which reduces to
$l_{\Gamma'}[c_{\Gamma,\Gamma'}/x]$, which equals $\dia\,c_{\Gamma}$ by
induction.

$l_{\Gamma',\sslock}[c_{\Gamma,\Gamma',\sslock}/x]=\letd{x}{\dia\,
c_{\Gamma,\Gamma'}}{\open\,\letd{x}{l_{\Gamma}}{\shut\,\dia\,x}}$. This reduces
to $\open\,\letd{x}{l_{\Gamma}[c_{\Gamma,\Gamma'}/x]}{\shut\,\dia\,x}$, which
by induction is equal to $\open\,\letd{x}{\dia\,c_{\Gamma}}{\shut\,\dia\,x}$. This reduces to
$\open\,\shut\,\dia\,c_{\Gamma}\mapsto\dia\,c_{\Gamma}$.


\section{Intuitionistic R}

\textbf{Proof of Lemma~\ref{q_preserved}:}

We first show that $\Box q\circ\eta^m=r:\Dia A\to\Box\Dia A$:
\[\xymatrix{
\Dia A \ar[r]^-r \ar[d]_{\eta^m} & \Box\Dia A \ar[d]_{\eta^m} \ar[dr]^-{id} \\
\Box\Dia\Dia A \ar[r]_-{\Box\Dia r} & \Box\Dia\Box\Dia A \ar[r]_-{\Box\varepsilon^m} & \Box\Dia A
}\]
Note that the bottom arrow is the definition of $\Box q$.

We then show that $\Box\Dia q\circ\eta^m=r$ also:
\[\xymatrix{
  & & \Box\Dia\Box A \ar[d]_{\Dia r=\Dia\Box r} \ar[dr]^-{\varepsilon^m} \\
  \Dia A \ar[r]_-{\Dia r} \ar[d]_{\eta^m} \ar[urr]^-{\Dia\eta^m} & \Dia\Box A  \ar[r]_-{\Dia\Box\eta^m} \ar[d]_{\eta^m} & \Dia\Box\Box\Dia A \ar[d]_{\varepsilon^m} & \Dia A \ar[dl]^-r \\
  \Box\Dia\Dia A \ar[r]_-{\Box\Dia\Dia r} & \Box\Dia\Dia\Box A \ar[r]_-{\Box\Dia\varepsilon^m} & \Box\Dia A
}\]
The right hand square commutes because $\Box\Dia\varepsilon^m\circ\eta^m=
\eta^m\circ\varepsilon^m=\varepsilon^m\circ\Dia\Box\eta^m$. The
top triangle commutes by the naturality of $r$; we then use that $\Box$ preserves
$r$ as shown. $\varepsilon^m\circ\Dia\eta^m=id$ completes the proof.

This establishes that $\Box q\circ\eta^m=\Box\Dia q\circ\eta^m$. But by
the adjunction there should be a unique arrow $h$ such that $\Box h\circ\eta^m=
r$, so $q=\Dia q$.

\textbf{Categorical Semantics:} Note that that for the soundness of $\beta$-reduction
for function spaces and $\Dia$ we need updated versions of
Lemmas~\ref{lem:catsubs} and~\ref{lem:cat_leftweak}; these are straightforward
from the naturality of weakening.

There are two simple lemmas we need for the soundness of $\beta$-reduction for
$\Box$ and $\Dia$:

\begin{lemma}\label{lem:weak_lem}
\[\xymatrix{
 \Gamma_1,\Gamma_2,\Gamma_3,\Gamma_4 \ar[rr]^-{w_{\Gamma_2,\Gamma_3,\Gamma_4}} \ar[dr]_-{\Gamma_4(w_{\Gamma_3})} & & \Gamma_1 \\
 & \Gamma_1,\Gamma_2,\Gamma_4 \ar[ur]_-{w_{\Gamma_2,\Gamma_4}}
}\]
\end{lemma}
\begin{proof}
By induction on $\Gamma_4$, with the base case using induction on $\Gamma_3$.
\end{proof}

\begin{lemma}
$\den{\Gamma,\Gamma',\Gamma''\vdash t:A}=\den{\Gamma,\Gamma''\vdash t:A}
\circ\den{\Gamma''}(w_{\Gamma'})$.
\end{lemma}
\begin{proof}
By induction on the formation of $t$. The variable case uses
Lemma~\ref{lem:weak_lem}.
\end{proof}

\textbf{Proof of Lemma~\ref{lem:weak_and_tick}:}

We use induction on $\Gamma'$. The base case is exactly
Lemma~\ref{q_preserved}.

$\Dia w_{\sslock,\Gamma',A}=\Dia w_{\sslock,\Gamma'}\circ\Dia pr=
w_{\Gamma',\sslock}\circ\Dia pr$ by induction. This is $w_{\Gamma'}\circ q\circ
\Dia pr=w_{\Gamma'}\circ pr\circ q$ as required by the naturality of $q$. The
step case with $\lock$ follows similarly.

\textbf{Proof of Lemma~\ref{lem:weak_and_ctx}:}

We first show the denotation of $w$ in the term model:
\begin{align*}
  w_{\cdot} &= x \\
  w_{\Gamma,A} &= w_{\Gamma}[\pi_1\,x/x] \\
  w_{\Gamma,\sslock} &= w_{\Gamma}[\letd{x}{x}{x}/x]
\end{align*}
The base case of the lemma is trivial. 

$w_{\Gamma',A}[c_{\Gamma,\Gamma',y:A}/x]=w_{\Gamma'}[\pi_1
\langle c_{\Gamma,\Gamma'},y\rangle/x]\mapsto
w_{\Gamma'}[c_{\Gamma,\Gamma'}/x]$, then apply induction.

$w_{\Gamma',\sslock}[c_{\Gamma,\Gamma',\sslock}/x]=w_{\Gamma'}
[\letd{x}{\early\,c_{\Gamma,\Gamma'}}{x}/x]\mapsto
w_{\Gamma'}[c_{\Gamma,\Gamma'}/x]$.

\end{document}